\documentclass[prx,
aps,
twocolumn,
superscriptaddress,
reprint,
floatfix,
nofootinbib
]{revtex4-2}

\pdfoutput=1
\usepackage{dsfont,amsmath,amssymb,amsthm, bbm,graphicx,mathtools,bm,ulem,physics}
\usepackage[dvipsnames]{xcolor}
\usepackage[toc,page]{appendix}
\usepackage{enumitem}
\usepackage{thmtools,thm-restate}
\usepackage[english]{babel}


\usepackage[most]{tcolorbox}

\newtcbox{\mymath}[1][]{ 
    nobeforeafter, math upper, tcbox raise base,
    enhanced, colframe=blue!10,
    colback=blue!10, boxrule=1pt,
    #1}

\newtcolorbox{bluebox}{breakable,nobeforeafter,tcbox raise base,enhanced, colframe=blue!10,
    colback=blue!10, boxrule=1pt,}
    
\newtcolorbox{greenbox}{breakable,nobeforeafter,tcbox raise base,enhanced, colframe=green!10,
    colback=green!10, boxrule=1pt,}
    
\newtcolorbox{redbox}{breakable,nobeforeafter,tcbox raise base,enhanced, colframe=red!10,
    colback=red!10, boxrule=1pt,}
    
\newtcolorbox{greybox}{breakable,nobeforeafter,tcbox raise base,enhanced, colframe=black!5,
    colback=black!5, boxrule=1pt,}

\def\Gmaj{\stackrel{G}{\rightarrow}}

\usepackage[colorlinks=true,linkcolor=blue,citecolor=blue,urlcolor=blue]{hyperref}

\renewcommand{\eqref}[1]{Eq.~(\ref{#1})}
\newcommand{\figref}[1]{Fig.~(\ref{#1})}

\newcommand{\secref}[1]{Section~\ref{#1}}
\newcommand{\appref}[1]{Appendix~\ref{#1}}

\newcommand{\lemref}[1]{Lemma~\ref{#1}}
\newcommand{\thmref}[1]{Theorem~\ref{#1}}

\newtheorem{proposition}{Proposition}
\newtheorem{theorem}{Theorem}
\newtheorem{lemma}{Lemma}
\newtheorem{corollary}{Corollary}

\theoremstyle{definition} 
\newtheorem{example}{Example}

\def\id{ {\mathbbm 1} }

\def\D{ {\mathcal D} }

\def\E{ {\mathcal E} }

\def\H{ {\mathcal H} }

\def\J{ {\mathcal J} }
\def\U{ {\mathcal U} }
\def\R{ {\mathcal R} }

\def\N{ {\mathcal N} }
\def\F{ {\mathcal F} }

\def\B{ {\mathcal B} }
\def\O{ {\mathcal O} }
\def\D{ {\mathcal D} }

\def\V{ {\mathcal V} }

\def\bms{ {\bm{s}} }

\def\maj{  \stackrel{\O}{\rightarrow} }

\definecolor{darkgreen}{rgb}{0.0, 0.5, 0.0}

\begin{document}

\title{A new approximate Eastin-Knill theorem }
\author{Rhea Alexander }
\email{ralexander@ugr.es} 
\affiliation{Departamento de Electromagnetismo y Física de la Materia, Universidad de Granada, 18071 Granada, Spain.}

\begin{abstract}

Transversal encoded gatesets are highly desirable for fault tolerant quantum computing. However, a quantum error correcting code which exactly corrects for local erasure noise and supports a universal set of transversal gates is ruled out by the Eastin-Knill theorem. Here we provide a new approximate Eastin-Knill theorem for the single-shot regime when we allow for some probability of error in the decoding. In particular, we show that a quantum error correcting code can support a universal set of transversal gates and approximately correct for local erasure if and only if the conditional min-entropy of the Choi state of the encoding and noise channel is upper bounded by a simple function of the worst-case error probability. Our no-go theorem can be computed by solving a semidefinite program, and, in the spirit of the original Eastin-Knill theorem, is formulated in terms of a condition that is both necessary and sufficient, ensuring achievability whenever it is passed.
As an example, we find that with $n=100$ physical qutrits we can encode $k=1$ logical qubit in the $W$-state code, which admits a universal transversal set of gates and corrects for single subsystem erasure with error probability of $\varepsilon = 0.005$. 
To establish our no-go result, we leverage tools from the resource theory of asymmetry, where, in the single-shot regime, a single (output state-dependent) resource monotone governs all state purifications.

\end{abstract}
    
\maketitle

\section{Introduction}

Recent experimental and theoretical advances are beginning to push the frontier of quantum computing closer to realization~\cite{Xu2024,Yamasaki2024Overhead,Acharya2025}. However, engineering scalable quantum architectures remains a formidable challenge, in part due to the significant overhead requirements associated with best known schemes for storing and processing quantum information in the fault-tolerant regime~\cite{Campbell2017Roads}.

Quantum error correcting codes that support transversal logical gates provide the most straightforward pathway to fault tolerance~\cite{Gottesman2006QECFT}.  By definition, a gate is transversal for a given code if it can be implemented at the physical level by applying unitaries which do not introduce entanglement within a given block of code. This intuitively limits errors from compounding in an uncontrollable manner. However, famously, the Eastin-Knill theorem~\cite{EastinKnill2009Restrictions} rules out the existence of a quantum error correcting code which both corrects for local erasure errors and supports a dense set of transversal gates, such the full unitary group.

There are a range of proposals for realizing universal fault tolerant gatesets~\cite{original_magic_states,Paetznick2013TransversalEC,OConnor2014Concatenated,Bombin2015GaugeFixing,Yoder2016Universal}, each of which can be viewed as making a particular choice regarding how to circumvent the Eastin-Knill theorem. One such approach is the magic state injection model~\cite{original_magic_states}, where transversality is restricted to gates generated by the discrete Clifford group, and universality is restored by the addition of a special class of ancillary states known as magic states. Despite recent promising advances in this direction~\cite{wills2024constant,Gidney2024MagicCultivation}, the overhead incurred by preparing high-quality magic states in the non-asymptotic regime is currently impractical.

An alternative route to achieving a universal set of fault-tolerant gates is to relax the requirement for exact quantum error correction~\cite{Woods2020continuousgroupsof,Faist2020ContinuousSymmQEC} by allowing for some probability of error in the decoding: the so-called regime of approximate quantum error correction~\cite{Leung1997Approximate,Oreshkov2010AppoximateQEC}. The hope being that, if this error can be sufficiently suppressed, then we can achieve a universal transversal set of gates without any additional overhead incurred by preparing high-quality magic states~\cite{original_magic_states}, or the need for intermediate rounds of error correction~\cite{Paetznick2013TransversalEC,OConnor2014Concatenated,Bombin2015GaugeFixing,Yoder2016Universal}.  

At its heart, the Eastin-Knill theorem expresses an inherent tension between quantum error correction and continuous symmetry principles~\cite{EastinKnill2009Restrictions}. Taking advantage of this connection, several recent studies have developed approximate variants of the
Eastin-Knill theorem~\cite{Woods2020continuousgroupsof,Faist2020ContinuousSymmQEC,Kubica2021ApproximateEK,Gupta2024Transversal}, each of which provides a robust variant of the original theorem by also encompassing the regime of finite error probability. More broadly, the study of covariant codes~\cite{Zhou2021newperspectives} finds relevance in diverse fields of physics beyond the study of fault tolerant quantum computing, such as in the context of quantum gravity~\cite{Almheiri2015AdS,Pastawski2015,Harlow2021QuantumGravity} and condensed matter physics~\cite{Brandao2019QECSpinChains}.

To date, most existing approximate Eastin-Knill theorems~\cite{Woods2020continuousgroupsof,Faist2020ContinuousSymmQEC,Kubica2021ApproximateEK} provide lower bounds on the error probability of unitarily covariant codes, and form necessary conditions on the existence of such codes. Meanwhile, there also exist several explicit constructions of $U(d)$-covariant codes~\cite{Yang2022Optimal,Kong2022NearOptCovariant}, each of which provide a specific upper bound on the (typically asymptotic) error probability. More recently, a very general no-go theorem on fault-tolerant encoded gatesets has been developed based on Lie algebraic properties of sets of correctable errors~\cite{Gupta2024Transversal}. Here we present a new approximate Eastin-Knill theorem in terms of a single necessary and sufficient condition, which can be viewed as complementary to these prior works. The benefit of this condition is that it is remarkably simple, and provides rigorous guarantees on whether or not a given code can operate below some desired error threshold in the presence of erasure noise.

The remainder of this paper is structured as follows. First, in \secref{sec:preliminaries} we establish some convenient notation and review some basic concepts from quantum error correction. In particular, we draw attention to \secref{sec:main_results} where we summarize the main results of this work. By drawing on techniques from the resource theory framework, in \secref{sec:asymmetry_purification} we construct the key tools that we will need to prove our main result. In \secref{sec:QEC} we prove our approximate Eastin-Knill theorem and consider a simple example of this theorem in action. Finally, in \secref{sec:outlook}, we conclude with a discussion of future directions for which we hope this work inspires.

\section{Preliminaries}
\label{sec:preliminaries}

\subsection{Notation}

Let us begin by briefly establishing some convenient notation.
Throughout we shall denote quantum systems by latin capital letters such as $A$, $B$, and $S$, with corresponding Hilbert spaces $\H_A$, $\H_B$, and $\H_S$, respectively. Throughout we restrict our attention to the consideration of Hilbert spaces which are finite-dimensional, namely, $d_S < \infty$ etc., where $d_S$ denotes the dimension of the Hilbert space $\H_S$. Accordingly, our results hold only within this finite-dimensional setting\footnote{One known way of circumventing the Eastin-Knill theorem~\cite{Hayden2021RefFrame} is by using infinite-dimensional codes.}. We will denote by $\B(S)$ the set of all bounded linear operators on the Hilbert space $\H_S$. We denote by $\D(S)$ the set of all quantum states (positive semidefinite, trace-one operators) on the Hilbert space $\H_S$. Moreover, we will use the notation $\mathrm{pure}(S)$ to denote the subset of $\D(S)$ consisting of all rank-1, pure states. When we want to clarify the input and output space of any quantum channel $\E : \B(A) \rightarrow \B(B)$ we will use $\E_{A\rightarrow B}$, and furthermore $\E_A \coloneq \E_{A \rightarrow A}$. Finally, we shall use the notation $[n] \coloneqq \{ 1, \dots, n\}$, for any integer $n$.

\subsection{Quantum codes with transversal gatesets}

Quantum error-correcting codes are an essential ingredient for building a quantum computer which can function in the presence of inevitable noise. The basic idea is to encode logical quantum information by mapping it to a subspace of a larger dimensional `physical' space. More formally, we shall define a quantum error correcting code via the completely positive and trace-preserving (CPTP) map $\E_{L \rightarrow P}$ from the Hilbert space of logical information $\H_L$ to this physical Hilbert space $\H_P$, which consists of $n$ subsystems $P \coloneqq P_1 P_2 \dots P_n$. The subspace of $P$ which is image of the encoding map $\E_{L \rightarrow P}$ is called the \textit{code space} $\mathcal{C}$.

The physical system $P$ in general will be subject to noise, which can always be represented by some quantum channel $\N_{P \rightarrow P'}$ where $P'$ may or may not be distinct from $P$. Here we shall call a code $\E_{L \rightarrow P}$ \textit{$\varepsilon$-correcting} under the noise channel $\N_{P \rightarrow P'}$, if there exists a recovery channel $\D_{P' \rightarrow L}$ such that for all pure states $\psi_L$ of the logical system
\begin{align} \label{eq:epsilon_correctability}
  F(  \D \circ \N \circ \E( \psi_L), \psi_L) \ge 1 - \varepsilon,
\end{align}
where $F (\rho , \sigma) \coloneqq (\tr \sqrt{ \sqrt{\rho} \sigma \sqrt{\rho} } )^2$ is the \textit{fidelity} between the two states $\rho$ and $\sigma$. Since the fidelity is (non-jointly) concave with respect to each of its arguments~\cite{wilde2013quantum}, if $\E$ is $\varepsilon$-correctable with respect to the noise channel $\N$, then there exists a decoder $\D$ such that
\begin{align}
 F(  \D \circ \N \circ \E( \rho_L), \rho_L) \ge \tr(\rho_L^2)( 1 - \varepsilon),
\end{align}
for any mixed state of the logical system $\rho_L$. Therefore, the condition of $\varepsilon$-correctability in \eqref{eq:epsilon_correctability} guarantees a notion of approximate recovery for all states $\rho_L$ of the logical system which are almost pure, as measured by the function $\tr[(\cdot)^2]$.

To implement useful quantum computations, an additional key ingredient is the ability to reliably process quantum information at the encoded level.  One of the simplest proposals for implementing a given logical gate in a way that is manifestly fault tolerant is if that gate is \textit{transversal}. In particular, we say that the gate $\U_L \coloneqq U_L(\cdot) U_L^\dagger$ can be implemented transversally in the code $\E_{L \rightarrow P}$ if there exists any collection of unitaries $\U_{P_i}$ for all $i \in [n]$ such that 
\begin{align} \label{eq:commute}
   \E_{L \rightarrow P} \circ \U_L = \bigotimes_{i = 1}^n \U^i_{P_i} \circ \E_{L \rightarrow P},
\end{align}
where the partitioning of $P$ into subsystems $P_i$ contains at most one physical subsystem (e.g. qubit) from a given block of quantum code. Ideally, we would like the transversality condition to hold for a universal set of gates, which is equivalent~\cite{Faist2020ContinuousSymmQEC} to \eqref{eq:commute} holding for all unitaries $U_L$ in $U(d_L)$, the group of unitaries of dimension $d_L$. Importantly, this condition is equivalent to the channel $\E$ being covariant with respect to the group $U(d_L)$. However, famously such a universal, transversal encoding which also can \textit{perfectly} correct local erasure errors, is ruled out by the Eastin-Knill theorem~\cite{EastinKnill2009Restrictions}.

Before proceeding, a quick note on our definition of $\varepsilon$-correctability compared to usage in the literature is in order. 
The definition of $\varepsilon$-correctability we consider here in \eqref{eq:epsilon_correctability} guarantees that a worst-case error of $\varepsilon$ is achieved when the decoding is applied \textit{locally} to any pure state of the logical system. This is a weaker constraint than that imposed by requiring that some worst-case error must be achieved whenever the code $\E_L$ is applied to part of some entangled state $\rho_{AL}$ over the bipartite system composed of the logical subsystem $L$ and any additional subsystem $A$. On the other hand, in general \eqref{eq:epsilon_correctability} enforces a stronger notion of $\varepsilon$-correctability than the average-case error $ \bar{\varepsilon} \coloneqq \max_{\D} \int_{U(d_L)} d g F(\D\circ \N \circ \E (\psi^g_L),\psi^g_L)$~\cite{Gilchrist2005Distance}, where the integral is performed over the uniform Haar measure $dg$ and where $\ket{\psi^g} \coloneqq g \ket{\psi}$ for some pure state vector $\ket{\psi}$ of the logical system. The average case error is related to the \textit{Choi error} $\varepsilon_{\mathrm{Choi}}$, which maximizes the optimal fidelity of the Choi state of the channel $(\D \circ \N \circ \E)_L$ with the Choi state of the identity channel on the logical system with respect to all possible decoders $\D$, via $\bar{\varepsilon} = \sqrt{\frac{d_L +1}{d_L}}\varepsilon_{\mathrm{Choi}}$~\cite{Horodecki1999Singlet,Nielsen2002Fidelity}. In fact, our main result (Corollary~\ref{cor:approx_eastin_knill}) indicates that an additional relation between the worst-case error in \eqref{eq:epsilon_correctability} and the Choi error holds, whenever we restrict our attention to $U(d)$-covariant channels. However, we leave the formalization of such a statement to future work.

\subsection{Main results}
\label{sec:main_results}

Our main contribution is to provide a new approximate Eastin-Knill theorem in terms of a single necessary and sufficient condition. We show that the code $\E_{L \rightarrow P}$ admits a transversal implementation of the full unitary group and is $\varepsilon$-correctable with respect to erasure of $m$ subsystems $\N_{P \rightarrow P'} \coloneqq \tr_{P_1 \dots P_m}$ if and only if 
\begin{align} \label{eq:mainresults_approx_ek}
  H_{\mathrm{min}}(L|P')_{J(\N \circ \E)} \le  - \log d_L(1 -c \varepsilon),
\end{align}
where $c  \coloneqq \frac{d_L + 1}{d_L}$, $ J (\E)$ is the Choi state of the channel $\E$, and $H_{\mathrm{min}}(A|B)_\Omega$ is the \textit{conditional min-entropy}~\cite{renner2005security} of the bipartite state $\Omega_{AB}$. By noting that the global minimum value of the conditional min-entropy with respect to any bipartite state over $LP$ is $-\log d_L$ (e.g.~\cite{tomamichel2013thesis}) which is achieved for maximally entangled states~\cite{Gour2024Inevitability}, we can immediately get a flavour for how \eqref{eq:mainresults_approx_ek} encodes the exact Eastin-Knill theorem~\cite{EastinKnill2009Restrictions}. Intuitively, this follows as a consequence of the fact that local subsystem erasure is an entanglement breaking channel, and hence we do not expect this optimum to be achieved.

Our entropic formulation of the approximate Eastin-Knill theorem expressed in \eqref{eq:mainresults_approx_ek} is remarkably simple, and moreover can be computed by solving a semidefinite program (SDP). Therefore, it can be solved efficiently for sufficiently low-dimensional systems. As an example, we show that for the case of the $U(d)$-covariant $W$-state code~\cite{Faist2020ContinuousSymmQEC}, this SDP can be reformulated as a simple analytic expression. In particular, we find that the W-state code, which admits a transversal implementation of the full unitary group, can $\varepsilon$-correct for the erasure of $N_e$ subsystems if and only if
    \begin{align} 
    \varepsilon \ge \frac{N_e}{n} \left( 1 - \frac{1}{d_L}\right),
\end{align}
where $n$ is the number of physical subsystems and $d_L$ is the dimension of the logical Hilbert space.
For example, this condition informs us that with $n=100$ physical qutrits we can encode $1$ logical qubit in the $W$-state code, which admits a universal transversal set of gates and corrects for single subsystem erasure with error probability of $\varepsilon = 0.005$. Finally, we construct an explicit optimal decoder for the $W$-state code for the case of known erasure.
\begin{figure}[t]
	\centering	\includegraphics[width=0.95\linewidth]{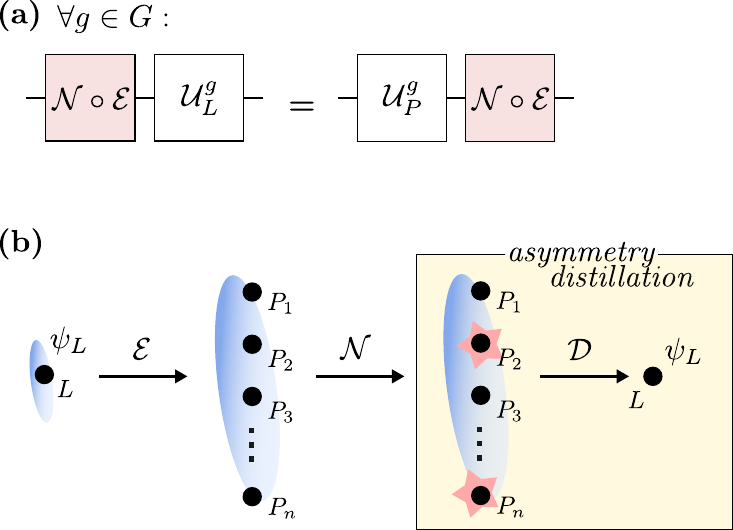}
\caption{\textbf{(Quantum error correction as asymmetry distillation).} Given any $G$-covariant encoder $\E_{L \rightarrow P}$ and $G$-covariant noise channel $\N_{P \rightarrow P}$, then their sequence $(N \circ \E)_{L \rightarrow P}$, as depicted in \textbf{(a)}, is also $G$-covariant.  \textbf{(b)} Whenever $(N \circ \E)_{L \rightarrow P}$ is indeed a $G$-covariant channel, as was first shown in Ref.~\cite{Zhou2021newperspectives},  the optimal decoder $\D_{P \rightarrow L}$ can also be assumed to be $G$-covariant (also see \lemref{lem:covariant_decoder_optimal}). Therefore, we can view quantum error correction as equivalent to asymmetry distillation: a fundamental task in the resource theory of asymmetry.
		\label{fig:QEC_asymmetry_distillation}}
\end{figure} 

To prove these results we leverage tools from the resource theory of asymmetry~\cite{Gour2008ResourceReference,Marvian2012PhdThesis}, a.k.a. \textit{asymmetry theory}. In particular, by drawing on prior work \cite{gour2018quantum,Marvian2020CoherenceDistillation} we present a series of purification theorems, which can be viewed as generalizing Noether's ``first-law-like" theorem~\cite{noether1918invarianten} for unitary quantum theory to ``second-law-like" theorems in purification scenarios where the symmetry-constrained dynamics in question are not necessarily unitary. In particular, for the case where the output state of a covariant transformation is pure, or approximately pure, it turns out that a single conditional min-entropy 
characterizes all allowed purifications within the resource theory of asymmetry in the single-shot regime~\cite{Marvian2020CoherenceDistillation}. To apply these results to the current setting, we generalize to the case where we require the purification procedure to work when applied to a (continuous) set of states. As was proven in Ref.~\cite{Zhou2021newperspectives}, and is shown schematically in \figref{fig:QEC_asymmetry_distillation}, given appropriate choices of group $G$ and unitary representations of $G$, quantum error correction can always be viewed as some form of asymmetry purification.

\section{Asymmetry purification}
\label{sec:asymmetry_purification}

The Eastin-Knill theorem~\cite{EastinKnill2009Restrictions} expresses an inherent tension between symmetry principles and exact quantum error correction. As such, it is natural to exploit tools from asymmetry theory~\cite{Gour2008ResourceReference,Marvian2012PhdThesis}, where symmetry breaking is viewed as a consumable resource.
In this section, we provide a brief introduction to the resource theory of asymmetry in the particular context of purification protocols. That is, those processes for which the desired output state of the procedure is pure or approximately pure, which, under certain conditions on the noise~\cite{Zhou2021newperspectives}, is precisely the setting of (covariant) quantum error correction.

\subsection{The resource theory of asymmetry}

At the highest level a resource theory~\cite{Chitambar2019QRT_review} assigns some collection of states $\F$ and operations $\O$ as free or readily available. In contrast, any state or indeed operation not belonging to one of these sets in some sense will constitute a resource with respect to the theory. Typically, these free sets are chosen to correspond to some fundamental limitation of nature or otherwise some practical constraint, allowing for a formal treatment of that restriction (e.g. see the recent review~\cite{gour2024resources}). 

One such resource theory of interest here is the resource theory of asymmetry~\cite{Gour2008ResourceReference,Marvian2012PhdThesis}, where the resource in question is the ability to break symmetry. 
More formally, given some compact symmetry group $G$ and adjoint unitary representation $g\rightarrow \U^g$ on $\B(\H)$, one can always define a corresponding resource theory of $G$-asymmetry~\cite{gour2008refframes,marvian2013theory}. The set of free states of this theory $\F_{G}$ are defined as the collection of all states which are invariant under the group action, namely, those states $\rho$ for which
\begin{align}
  \U_g (\rho) = \rho,
\end{align}
for all $g \in G$.
Any free state $\rho \in  \F_{G}$ is a symmetric state with respect to the group $G$, and conversely the resource states of the theory are those which are not invariant under the group action. The free operations $\O_{G}$ of this theory are the set of all channels which commute with the group action, called the set of \textit{$G$-covariant} channels. More formally, any $G$-covariant channel $\E$ from system $A$ to system $B$ satisfies
\begin{align} \label{eq:G-cov_channels}
   \E_{A \rightarrow B} \circ \U^g_A = \U^g_B \circ \E_{A \rightarrow B} , 
\end{align}
for all $g \in G$. Implicitly, any $G$-covariant channel is actually a $\left(G, \{\U^g_{A}\}, \{\U^g_B\}\right)$-covariant channel, that is, it depends on the choice of representations of the group on the input and output spaces. However, this notation is rather cumbersome, and so is typically kept implicit without causing too much confusion.
Like in any resource theory, a question of central concern within the resource theory of asymmetry is as follows: given any two states $\rho_A$ and $\sigma_B$ of systems $A$ and $B$, does there exist a free operation $\E_{A \rightarrow B} \in \O_G$ such that $\E(\rho) = \sigma$? Whenever such an operation does indeed exist, we shall write 
\begin{align} \label{eq:G_preorder}
    \rho \Gmaj \sigma.
\end{align}
Knowledge of the full pre-order in \eqref{eq:G_preorder} provides a full specification of the asymmetry properties of a quantum system. 

Recent work~\cite{gour2018quantum} gave the first set of complete but infinite set of asymmetry monotones which determine the pre-order in \eqref{eq:G_preorder}, which together are SDP computable. The set of monotones
 $H_\eta$ in question are constructed from the \textit{conditional
min-entropy}~\cite{renner2005security} via
\begin{align}
    H_\eta (\rho) &\coloneqq H_{\mathrm{min}}(R|A)_{\Sigma} \notag \\
    &= -\log \inf_{X \ge 0}\{\tr(X) \ | \ \id_R \otimes X_A - \Sigma_{RA} \ge 0 \}
\end{align}
specialized to the following bipartite, globally $G$-invariant state 
\begin{align}
    \Sigma_{RA} &\coloneqq \Pi^G_{RA}(\eta_R^T \otimes \rho_A) \notag \\ &= \int_G dg \, \bar{\U}^g_R \otimes \U^g_A (\eta_R^T \otimes \rho_A) 
\end{align}
where $\Pi^G_{RA}$ is the global \textit{G-twirl} on the space $\B(\H_{RA})$ and where the notation $\bar{\cdot}$ denotes complex conjugation with respect to a fixed orthonormal basis. In particular, $\rho_A \Gmaj \sigma_B$ if and only if~\cite{gour2018quantum}
\begin{align}
    H_\eta (\rho) \le H_\eta (\sigma),
\end{align}
for all possible states $\eta$ of the system $R$, where $R \cong B$. As we shall see however, when we restrict our attention to the task of resource purification, in fact only one such monotone suffices.

\subsection{Exact purification}
\label{sec:exact_symmetry_principles}

To warm up, let us consider the setting of exact purification, where our central question is whether or not there exists a $G$-covariant map $\E$ which exactly maps some arbitrary state $\rho$ to some pure state $\psi$. In this context, there exists a single monotone which provides a complete characterization of this problem, as we formalize in the following proposition.

\begin{proposition}[\cite{Marvian2020CoherenceDistillation}] \label{prop:asymmetry}

Let us consider two quantum systems $A$, $B$, and two states $\rho_A \in \D(A)$ and $\psi_B \in \mathrm{pure}(B)$. There exists a $G$-covariant channel mapping $\rho$ to $\psi$ if and only if
    \begin{align}
      H_{\psi_{B}} (\rho_A) \le  H_{\psi_{B}} (\psi_{B}),
    \end{align}
 where we have defined
\begin{align}
    H_{\psi_R} (\rho_S) \coloneqq H_{\min}(R|S)_{\Pi^G(\psi_R^T \otimes \rho_S)}.
\end{align}
\end{proposition}

This theorem follows as a consequence of an identity initially proven in Supplementary Material 9 of Ref.~\cite{Marvian2020CoherenceDistillation}. A proof can also be found in \appref{appx:proofs_of_theorems}.

Proposition~\ref{prop:asymmetry} offers a quadratic computational reduction of the SDP search space from $(d_A d_B)^2$ to $d_A d_B$ compared to the complete SDP conditions for asymmetry theory presented in Ref.~\cite{gour2018quantum}. However, the real benefit of this condition lies in it's conceptual simplicity. That is, for transitions to any rank-one pure state $\psi$ this proposition provides a \textit{single entropic quantity} which determines the state conversion problem within the resource theory of asymmetry. This is rather nice, as single resource monotones are few and far between in quantum resource theories~\cite{Datta2023FiniteMonotones}.

One potential way of interpreting Proposition~\ref{prop:asymmetry} can be borrowed from Ref.~\cite{gour2018quantum} in terms of quantum reference frames. In particular, for the case of pure state conversion it tells us that an initial state $\ket{\psi^i}$ can be $G$-covariantly transformed into some final state $\ket{\psi^f}$ if and only if the ``time-reversed" version of $\ket{\psi^f}$ serves as a better reference frame for $\ket{\psi^i}$ than it does for itself.

 Let us consider an example wherein Proposition~\ref{prop:asymmetry} can be applied.

\begin{tcolorbox}[breakable,colback=darkgreen!5!white,colframe=darkgreen!0!white]
\begin{example}[Entropic constraints on evolution of the universe]
Let us consider the quantum system $S$ corresponding to the universe with Hamiltonian $H_S$ and an additional system $\tilde{S} \cong S$. Moreover, let us note the following axioms:
\begin{enumerate}
    \item[I.]  \label{axiom1} The state of the system $S$ of the universe at some initial time labelled by $t=0$ is described by a pure quantum state $\psi^{i}$.
    \item[II.] \label{axiom2} Quantum mechanics holds globally across $S$.
    \item[III.] \label{axiom3} The global dynamics of the universe are covariant with respect to time-translations.
\end{enumerate} If axioms I-III hold, then it follows immediately from Proposition~\ref{prop:asymmetry} that some state $\psi^{f}$ is a possible state of the universe $S$ for any later time $t>0$ if and only if 
\begin{align}
   H_{\psi^{f }_{\tilde{S}}}(\psi^f_S) \ge H_{\psi^{f }_{\tilde{S}}}(\psi^i_S) ,
\end{align}
where here
\begin{align}
    H_{\psi} (\rho_S) \coloneqq H_{\min}(R|S)_{\Pi^{U(1)}(\psi_R^T \otimes \rho_S)}.
\end{align}
and 
\begin{align}
    \Pi^{U(1)}(\cdot) &\coloneqq \int_{-\infty}^{\infty} dt \, \bar{\U}_t\otimes \U_t(\cdot), \notag \\
    \U_t(\cdot)& \coloneqq e^{-iHt}(\cdot)e^{iHt}.
\end{align}
\end{example}
\end{tcolorbox}

This example gives a single necessary and sufficient condition for global state transitions of the universe, given some commonly held assumptions. The above example tells us that if one were (somehow) able to find some states of the universe $\psi(0)_S$ and $\psi(t)_S$ such that $t>0$ and
    \begin{align} \label{subeq:entropy_universes}
    H_{\psi(t)}(\psi(t)) <   H_{\psi(t)}(\psi(0)),
\end{align}
we would be lead to reject at least one of the axioms I-III.

\subsection{Approximate purification}
\label{sec:approximate_single_state}

Here we consider some initial state $\rho_A$, and we seek to answer whether the pure target state $\psi_B $ is achievable using only $G$-covariant operations, up to some error tolerance $\varepsilon$. To formalize this, we will say that a state $\sigma_S$ is \textit{$\varepsilon$-close} to the target state $\psi_S$ of the same system $S$ whenever $F(\sigma , \psi) \ge 1- \varepsilon$. Now we seek to determine whether there exists some $G$-covariant channel $\E_{A \rightarrow B}$ and some state $\sigma_B$ which is $\varepsilon$-close to our target state $\psi_B$, such that
\begin{align}
    \E(\rho_A) = \sigma_B.
\end{align}
Whenever such a $G$-covariant map exists, we shall write 
\begin{align}
    \rho \stackrel{G}{\rightarrow}_\varepsilon \psi.
\end{align}
In this context, we have the following proposition, which again follows as a consequence of an identity initially proven in Supplementary Note 9 of Ref.~\cite{Marvian2020CoherenceDistillation}.

\begin{proposition}[\cite{Marvian2020CoherenceDistillation}]  \label{prop:approximate}
Let $0 \le \varepsilon \le 1$. Moreover, let us consider two quantum systems $A$, $B$, and two states $\rho_A \in \D(A)$ and $\psi_B \in \mathrm{pure}(B)$. There exists a $G$-covariant channel mapping $\rho$ to $\psi$ up to error $\varepsilon$ if and only if \begin{align} \label{eq:approximate}
H_{\mathrm{min}}(B|A)_{\Pi^G(\psi_{B}^T \otimes \rho_A)} \le  \log \frac{1}{ 1- \varepsilon}.\end{align}
\end{proposition}

A proof of this theorem is provided in \appref{appx:proofs_of_theorems}.

The special case of this theorem with $\varepsilon=0$ is in fact equivalent to Proposition~\ref{prop:asymmetry}. To see this, note that $H_{\psi} (\psi)=0$ for any pure state $\psi$. With this in mind we can identify the term on the left-hand side of \eqref{eq:approximate} as an entropy production $\Delta H$ under the transition $\rho \rightarrow \psi$, namely
\begin{align}
    \Delta H &\coloneqq H_{\psi} (\rho) -H_{\psi} (\psi) =H_{\mathrm{min}}(B|A)_{\Pi^G(\psi_{B}^T \otimes \rho_A)}.
\end{align}

For finite values of the error parameter $\varepsilon$, the case where we take the input state $\rho$ to be a pure state $\psi^i$, where we consider transitions to pure states $\psi^f$ of the same system $A=B$, and the case where $G$ is a compact Lie group, this can be viewed as an approximate recasting of Noether's theorem~\cite{noether1918invarianten} to the case where some symmetry principle holds only approximately.

\subsection{Multi-state approximate purification}
\label{sec:multi-conversion}

Now let us generalize the known results we have seen so far to the case where we have a pair of ensembles $ \{\rho_\mu\}_{\mu \in \Lambda}$ and $  \{\psi_\mu\}_{\mu \in \Lambda}$ of states parameterized by some real-valued $\mu$ defined over some compact set $\Lambda$. Now we seek to answer the question of whether or not there exists a $G$-covariant quantum channel $\E $ such that
\begin{align}
    \E(\rho_\mu) = \sigma_\mu \text{ and }  F(\sigma_\mu,\psi_\mu) \le 1- \varepsilon,
\end{align}
for all $\mu \in \Lambda$. For example, $\Lambda$ could denote a two element set $\{0,1\}$ and the question becomes whether or not there exists a $G$-covariant channel $\E$ such that $\E( \rho_0 ) $ is $\varepsilon$-close to $\psi_0$ and $\E( \rho_1 ) $ is $\varepsilon$-close to $\psi_1$. More generally, the set $\Lambda$ could parameterize some continuum of states.

In this setting, we have the following new result.

\begin{lemma}
\label{lem:multi}
Let $0\le \varepsilon \le 1$. Moreover, let $ \{\rho_\mu\}_{\mu \in \Lambda}$ and $ \{\psi_\mu\}_{\mu \in \Lambda}$ be two sets of states such that $\rho_\mu \in \D(A)$ and $\psi_\mu \in \mathrm{pure}(B)$ for all $\mu \in \Lambda$, where $\Lambda$ is a compact set. There exists a $G$-covariant channel mapping $\rho_\mu$ to $\psi_\mu$ up to error $\varepsilon$ for all $\mu \in \Lambda$ if and only if 
\begin{align}
\max_{\{p(\mu)\}} H_{\mathrm{min}}(B|A)_{\Sigma_{BA}^p} \le \log\frac{1}{1- \varepsilon},
\end{align}
where
\begin{align}
\Sigma_{BA}^p &\coloneqq \Pi^{G} \left( \int_\Lambda d\mu \, p(\mu) \psi_\mu^T \otimes \rho_\mu \right),
\end{align}
and where the maximization is performed all probability density functions $ p(\mu) $ over $\Lambda$ such that $ \int_\Lambda d\mu \, p(\mu) = 1$ and $p(\mu) \in \mathbb{R}^+$ for all $\mu \in \Lambda$.
\end{lemma}

A proof of this lemma can be found in \appref{appx:proofs_of_theorems}.

In the next section, we will see the utility of the above result in the context of quantum error correction, where we seek decoders whose performance needs to be assessed when applied to a continuous set of states, rather than simply a single state of the logical system.

\section{Unitarily covariant quantum codes}
\label{sec:QEC}

Here, we leverage the resource-theoretic findings of the previous section to establish a new approximate Eastin-Knill theorem. The basic idea is sketched in \figref{fig:QEC_asymmetry_distillation} and we outline it as follows.
Given an encoding channel $\E_{L\rightarrow P}$ which is covariant with respect to some group $G$, and $G$-covariant noise channel $\N_{P \rightarrow P'}$, it is known \cite{Zhou2021newperspectives} that the optimal decoder $\D_{P' \rightarrow L}$ can also be assumed to be $G$-covariant. Therefore, in such cases the existence of such a decoder which operates up to infidelity $\varepsilon$ is equivalent to the problem of whether or not
\begin{align}
(\N \circ \E \circ \psi) \stackrel{G}{\rightarrow}_\varepsilon \psi,
\end{align}
for all pure states $\psi$ on the logical space. This was precisely the problem tackled in \secref{sec:multi-conversion}. By specializing to the group $G=U(d)$, here we find that \lemref{lem:multi} can be simplified considerably to provide a new approximate Eastin-Knill theorem.

\subsection{Approximate unitarily covariant codes}

 As shown in Lemma~2 of Ref.~\cite{Zhou2021newperspectives} under $G$-covariant encoding and noise, the optimal decoder can always be assumed to also be $G$-covariant. For completeness, in \appref{app:cov_decoder_optimality} we provide an analagous proof for our precise definition of $\varepsilon$-correcting codes given in \eqref{eq:epsilon_correctability}. 
By drawing on this result and the complete asymmetry distillation conditions given in \lemref{lem:multi} for multi-state asymmetry purification, we can now present the following necessary and sufficient condition for the existence of a recovery operation associated with a $U(d_L)$-covariant code $\E_{L \rightarrow P}$ which corrects for $U(d_L)$-covariant noise $\N_{P\rightarrow P'}$ up to error $\varepsilon$.

\begin{restatable}{theorem}{Udcovariant}
\label{thm:Ud_covariant} Let $\E_{L\rightarrow P}$ be an encoder which is $U(d_L)$-covariant (with respect to representations $\U^g_L$ and $\U^g_P$), and let $\N_{P \rightarrow P'}$ be a $U(d_L)$-covariant noise channel (with respect to representations $\U^g_P$ and $\U^g_{P'}$ on $P$ and $P'$). The encoder $\E_{L \rightarrow P}$ is $\varepsilon$-correctable under the noise $\N_{P \rightarrow P'}$ if and only if
\begin{align}
    H_{\mathrm{min}}(L|P')_{J(\N \circ \E)} \le  - \log d_L (1-c \, \varepsilon ),
\end{align}
where $c  \coloneqq \frac{d_L + 1}{d_L}$ and  $ J (\Phi_{L \rightarrow P'}) \coloneqq d_L^{-1} \mathrm{id}_L \otimes \Phi_{\tilde{L} \rightarrow P'}  \sum_{i,j} \ketbra{ii}{jj}_{L \tilde{L}}$,
denotes the Choi state of the channel $\Phi_{L \rightarrow P'}$.
\end{restatable}

A proof of \thmref{thm:Ud_covariant} can be found in  \appref{apx:proof_lemma_Ud_cov_encoder}.

\subsection{Covariant noise}

We seek to exploit \thmref{thm:Ud_covariant} to derive a new approximate Eastin-Knill theorem. To do so, we first need to establish which noise channels are in fact $U(d)$-covariant. To this effect, first let us define the following erasure channel $\N_{P \rightarrow P'}$, for $P= P_1 P_2 \dots P_n$ and $P'=P_1 \dots P_{j}$ 
\begin{align} \label{eq:erasure_noise}
    \N_{P \rightarrow P'}^{\mathrm{erasure}} (\cdot) \coloneqq \tr_{P_{j+1}, \dots, P_n}(\cdot),
\end{align}
which can be viewed as the known erasure channel considered in~\cite{Faist2020ContinuousSymmQEC} followed by discarding the information regarding which subsystem was erased. Without loss of generality we assume that the final $n-j-1$ subsystems were erased, since this amounts to simply a relabeling of qubits. On the other hand, we can define the partially depolarizing channel $\N_{P }^{(p)}$ for any probability-valued $p$ as
\begin{align} \label{eq:partially_depolarizing}
    \N_{P }^{(p)} (\cdot) \coloneqq (1-p)(\cdot) + p \frac{\id_P}{d_P} ,
\end{align}
which can be implemented by completely depolarizing with probability $p$ and doing nothing with probability $1-p$.

Importantly for us, the noise channels appearing in Eqs.~(\ref{eq:erasure_noise}) and (\ref{eq:partially_depolarizing}) are covariant with respect to any group $G$, and therefore we will see that \thmref{thm:Ud_covariant} immediately applies. Let us first formalize this in the following lemma.

\begin{lemma} \label{lemma:covariant_noise} Let us consider two systems $P=P_1 \dots P_n$ and $P' = P_1 \dots P_j$ of $n$ and $j$ subsystems respectively, and let $G$ be a compact group. We note the following:
 \begin{enumerate}
     \item \label{item:erasure} \textit{(Unknown erasure).} Here we have $j \le n$. Unknown erasure noise $\N_{P \rightarrow P'}^{\mathrm{erasure}}$ as given in \eqref{eq:erasure_noise} is covariant with respect to the group $G$ and tensor product unitary representations $\U_P^g\coloneqq \bigotimes_{i=1}^n \U_{P_i}^g$ and $\U_{P'}^g \coloneqq \bigotimes_{i=1}^j \U_{P_i}^g$ of $G$.
     \item \label{item:depolarizing} \textit{(Partially depolarizing noise).} Here we have $j=n$, that is, $P=P'$. Partially depolarizing noise $ \N_{P }^{(p)}$ as given in \eqref{eq:partially_depolarizing} is covariant with respect to the group $G$ and arbitrary unitary representation $\U_P^g$ of $G$, for all $p \in [0,1]$.
 \end{enumerate}   
\end{lemma}
\begin{proof}[Proof of \ref{item:erasure}]We first note the following commutation relation for the partial trace: $\tr_A \circ \U_{A} \otimes \V_B =  \V_B \circ \tr_A$, which holds for any pair of local unitary channels $\U_A$ and $\V_B$. Making use of this identity, we obtain for all $g \in G$
    \begin{align}
       \N_{P \rightarrow P'}^{\mathrm{erasure}} \circ \bigotimes_{i=1}^n \U_{P_i}^g &= \tr_{P_{j+1}, \dots, P_n} \circ \bigotimes_{i=1}^n \U_{P_i}^g \notag \\
       &= \bigotimes_{i=1}^j \U_{P_i}^g \circ \tr_{P_{j+1}, \dots, P_n} \notag \\
       &= \bigotimes_{i=1}^j \U_{P_i}^g \circ  \N_{P \rightarrow P'}^{\mathrm{erasure}}.
    \end{align}
This confirms that unknown erasure is indeed $G$-covariant with respect to adjoint unitary representations of tensor product form, completing the proof. \end{proof}
\begin{proof}[Proof of \ref{item:depolarizing}]
For any probability $p$ we straightforwardly have for all $g \in G$
\begin{align}
 \U^g_P \circ  \N_P^{(p)}(\cdot) &= (1-p )\U^g_P (\cdot) + p \frac{\id_P}{d_P} \notag  \\ &= \N_P^{(p)} \circ \U^g_P(\cdot),
\end{align}
and therefore partially depolarizing noise is indeed $G$-covariant with respect to arbitrary adjoint unitary representation on $\H_P$.\end{proof}

\subsection{Approximate Eastin-Knill theorem}

We now present our main result which provides a single necessary \textit{and sufficient} (and SDP computable) condition on the existence of approximate quantum error correcting codes supporting a universal transversal set of gates. In particular, combining \thmref{thm:Ud_covariant} and \lemref{lemma:covariant_noise} produces the following corollary. 

\begin{tcolorbox}[breakable,colback=blue!5!white,colframe=blue!0!white]
\begin{corollary}[Approximate Eastin-Knill theorem] \label{cor:approx_eastin_knill}  Let $P \coloneqq P_1 \dots P_n$ be a quantum system with $n$ subsystems, $P_{/j } \coloneqq P_1 \dots P_{j-1} P_{j+1} \dots P_n$, and let $L$ be another quantum system. Any code $\E_{L \rightarrow P}$ admits a transversal implementation of the full unitary group and is $\varepsilon$-correctable with respect to erasure of the $j$th subsystem if and only if 
\begin{align} \label{eq:approx_EK}
  H_{\mathrm{min}}(L|P_{/j}&)_{J(\tr_{P_j} \circ \E)} \le  - \log d_L(1 -c \varepsilon),
\end{align}
where $c  \coloneqq \frac{d_L + 1}{d_L}$ and  $ J (\Phi_{L \rightarrow P}) \coloneqq d_L^{-1} \mathrm{id}_L \otimes \Phi_{\tilde{L} \rightarrow P}  \sum_{i,j} \ketbra{ii}{jj}_{L \tilde{L}}$
denotes the Choi state of the channel $\Phi_{L \rightarrow P}$.
\end{corollary}
\end{tcolorbox}

\begin{proof}
By \lemref{lemma:covariant_noise}.1 erasure of the subsystem $P_j$, represented by the noise channel $\N_{P \rightarrow P_{/j}}\coloneqq \tr_{P_j}$ is $U(d_L)$-covariant with respect to the tensor product unitary channel representations $\V_P^g\coloneqq \bigotimes_{i=1}^n \V_{P_i}^g$ and $\V_{P'}^g \coloneqq \bigotimes_{i=1}^{n-1} \V_{P_i}^g$, where without loss of generality we have assumed that final subsystem is erased (solely to simplify notation). Let $\E_{L \rightarrow P}$ be a $U(d_L)$-covariant channel with respect to unitary representations $\V_L^g$ and \textit{the same} tensor product representation $\V_P^g$ of $U(d_L)$, such that 
\begin{align} \label{eq:products}
   \E_{L \rightarrow P} \circ \V_L^g =  \bigotimes_{i=1}^n \V_{P_i}^{g } \circ  \E_{L \rightarrow P},
\end{align}
for all $g \in U(d_L)$. \eqref{eq:products} encompasses the case where $\E_{L \rightarrow P}$ forms a transversal encoding of the full unitary group. Moreover, this is precisely the setting of \thmref{thm:Ud_covariant}. It follows immediately from \thmref{thm:Ud_covariant} that the channel $\E_{L \rightarrow P}$ is $\varepsilon$-correctable against the noise channel $\N_j \coloneqq \tr_{P_j}$ if and only if 
\begin{align} \label{subeq:approx_single_erasure}
 H_{\mathrm{min}}(L|P_{/j})_{J(\tr_{P_j} \circ \E)} \le - \log d_L(1 -c \varepsilon).
\end{align}
which is equivalent to the condition in \eqref{eq:approx_EK}, as claimed.
\end{proof}

We highlight that this result also can readily be formulated in terms of the erasure of any number $m<n$ physical subsystems, as the proof remains unchanged.

Let us see how Corallary~\ref{cor:approx_eastin_knill} encodes the exact Eastin-Knill theorem~\cite{EastinKnill2009Restrictions} whenever we have $\varepsilon = 0$. 
For any quantum state $\Omega_{LP'}$ we have the following global lower bound on the conditional min-entropy of this state (e.g. see Lemma 4.2 of Ref.~\cite{tomamichel2013thesis}):
\begin{align} \label{eq:hmin_lower_bound}
    H_{\mathrm{min}}(L|P')_{\Omega_{LP'}} \ge -\log d_L.
\end{align}
Now, the Choi state $J(\tr P_j \circ \E)$ is always a valid quantum state on $LP_{/j}$ for every $j \in [n]$. Therefore, combining the lower bound in \eqref{eq:hmin_lower_bound} with the upper bound in \eqref{eq:approx_EK} (with $\varepsilon =0$) we find that any code $\E_{L \rightarrow P}$ admits a transversal implementation of the full unitary group and can perfectly correct for the erasure of the $j$th subsystem if and only if
\begin{align}
    H_{\mathrm{min}}(L|P')_{J(\tr_{P_j} \circ \E)} = - \log d_L.
\end{align}
We can intuitively see that this reproduces the impossibility result of Eastin and Knill~\cite{EastinKnill2009Restrictions} for $d_L <\infty$\footnote{Throughout our analysis we have assumed that $U(d_L)$ is a compact group, for which we require $d_L < \infty$.}. 
In particular, if $\E$ corrects for local erasure errors and admits a universal transversal gateset, we would need $\max_{k \in [n]}  H_{\mathrm{min}}(L|P')_{J(\tr_{P_k} \circ \E)}$ equal to the global minimum value of $ H_{\mathrm{min}}(L|P')_{\Omega}$. This holds whenever $J(\tr_{P_k} \circ \E)$ is a maximally entangled state~(e.g. see \cite{Gour2024Inevitability}) for all $k \in [n]$, which we do not expect to hold since $\tr_{P_k}$ is an \textit{entanglement-breaking} channel.  

On the other hand, let us suppose the case of finite but very small $\varepsilon$. When we have an isometric encoder $\E_{L \rightarrow P}(\cdot) \coloneqq V_{L \rightarrow P} (\cdot) V^\dagger$ and the number of physical subsystems $n$ is very very large, we intuitively expect the condition in Corollary~\ref{cor:approx_eastin_knill} to pass. This is due to the fact that the conditional min entropy is invariant under local isometries, and furthermore the partial trace of one subsystem of a very large system isn't likely to break much of the global entanglement. This is in line with previous results in the literature, which have shown that one way to circumvent the Eastin-Knill theorem is to resort to infinite-dimensional codes~\cite{Hayden2021RefFrame,Faist2020ContinuousSymmQEC}.

\subsection{Example: W-state encoder}

Let us consider a simple example of the so-called $W$-state encoder first introduced in Ref.~\cite{Faist2020ContinuousSymmQEC}. This quantum error correcting code is manifestly covariant with respect to the full unitary group. In this code, the logical system of dimension $d \coloneqq d_L$ is mapped to a physical system with $n$ physical subsystems each of dimension $d+1$. In particular, for any pure state of the logical system $\ket{\psi}_L \coloneqq \sum_{i=0}^{d -1} c_i\ket{i}_L$, the encoding map is given by 
 $ \ket{\psi}_L \rightarrow \ket{\psi^{(n)}}_{P_1 \dots P_n}$, where
\begin{align}
\ket{\psi^{(n)}}  \coloneqq \frac{1}{\sqrt{n}} &(\ket{\psi,d, \cdots ,d } +   \notag \\ 
 &\ket{d,\psi,  \cdots ,d } +\dots  + \ket{d , \dots,d ,  \psi }),
\end{align}
where each physical subsystem forms a duplicate of the logical system with the extra basis vector $\ket{d}$.
Any logical gate $U$ can be implemented transversally at the encoded level by simply constructing the unitary $\tilde{U} \coloneqq U \oplus \id$, where $U$ acts on $\mathrm{span}\{\ket{i} \ | \ i = 1,\dots, d-1\}$ and $\id$ acts (as the identity) on $\mathrm{span}\{\ket{d}\}$) and applying $\tilde{U}^{\otimes n}$ to the encoded state $\ket{\psi^{(n)}}$.
\begin{figure}[t]
	\centering	\includegraphics[width=0.95\linewidth]{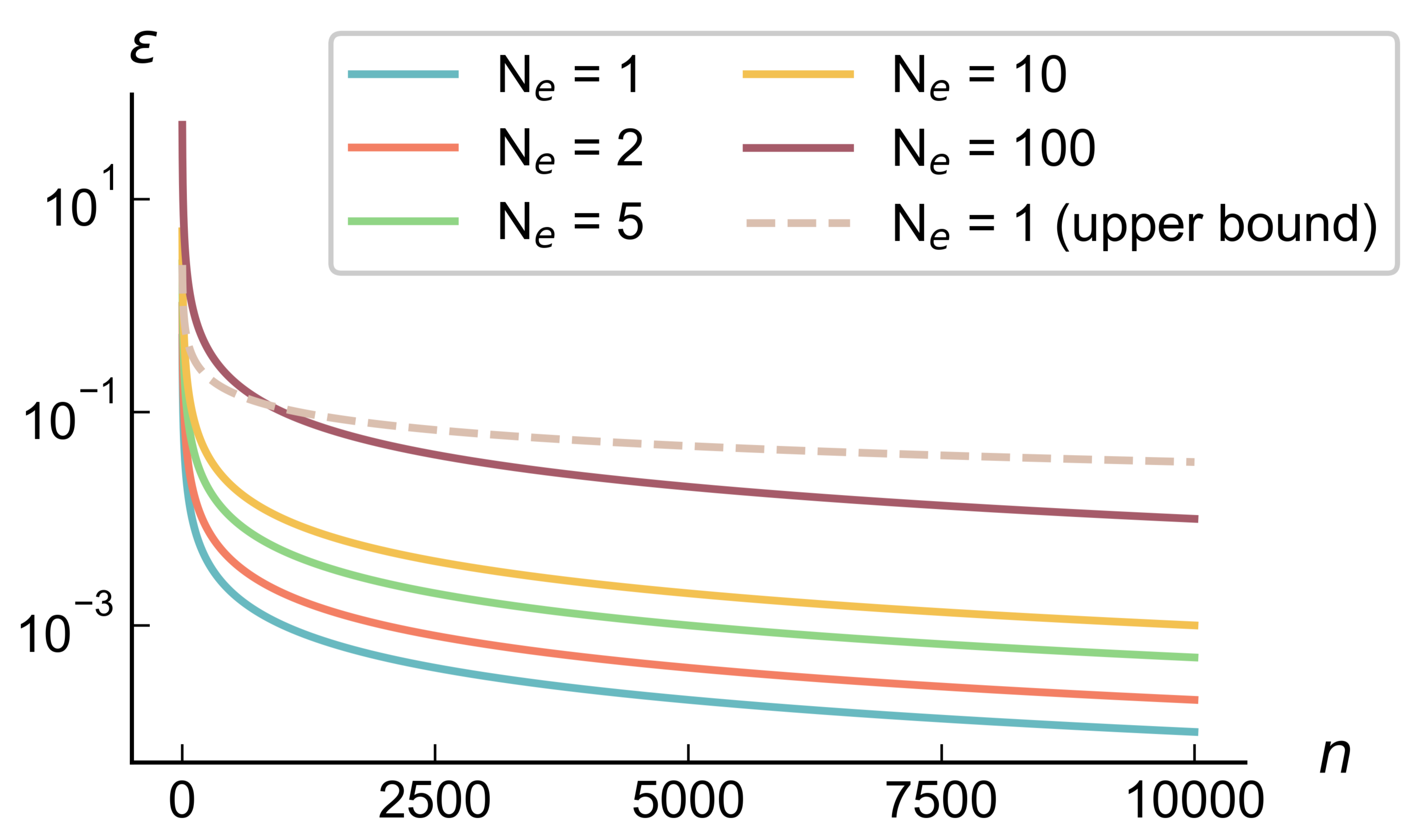}
\caption{\textbf{(Performance of the transversal W-state code)} The $W$-state code admits a transversal implementation of the full unitary group and can $\varepsilon$-correct for the erasure of $N_e$ physical subsystems. Here we plot the performance of this code with increasing number of physical subsystems $n$ in the limit where the number of encoded qubits $k = \log d_L \rightarrow \infty$ for different values of $N_e$, whose achievability is guaranteed by Corollary~\ref{corollary:W-code}. In the dashed line we plot the upper bound on the W-state code for the case of one encoded qubit ($k = 1$) provided in Ref.~\cite{Faist2020ContinuousSymmQEC}.
		\label{fig:Wstate}}
\end{figure} 

More explictly, the isometric encoder $\E^{(n)}$ for this code can be written as
\begin{equation}
\label{eq:W_state_encoder}
\E^{(n)}_{L \rightarrow P} \coloneqq V^{(n)}(\cdot) V^{(n) \dagger}, \ \ V^{(n)} \coloneqq \sum_i \ketbra{i^{(n)}}{i_L},
\end{equation}
where, for $i \in \{0,\dots, d-1 \}$,
\begin{align}
    \ket{i^{(n)}} \coloneqq \frac{1}{\sqrt{n}} &(\ket{i,d, \cdots ,d } +   \notag \\ 
 & \ket{d,i,  \cdots ,d } +\dots  + \ket{d , \dots,d ,  i }),
\end{align}
are the code words of the $W$-state code over $n$ physical qudits. For this code, which is $U(d)$-covariant, we have the following corollary of Corollary~\ref{cor:approx_eastin_knill}.
\begin{restatable}{corollary}{Wcode} \label{corollary:W-code} Let $n \in \mathbb{N}$. The W-state code $\E^{(n)}$ defined in \eqref{eq:W_state_encoder}, which admits a transversal implementation of the full unitary group, can $\varepsilon$-correct for the erasure of $N_e$ subsystems if and only if
    \begin{align} \label{eq:W_code_bound_epsilon}
    \varepsilon \ge \frac{N_e}{n} \left( 1 - \frac{1}{d_L}\right).
\end{align}
\end{restatable}

An explicit proof can be found in \appref{Appx:proof_corollary_W_code} but essentially amounts to substituting \eqref{eq:W_state_encoder} into \eqref{eq:approx_EK} of Corollary~\ref{cor:approx_eastin_knill}.

We highlight that this expression for approximate error correcting codes has a scaling for $\varepsilon$ with the number of physical qudits $n$ which, up to a constant factor in the range $\left[\frac{1}{2}, 1\right]$, is independent of the number of encoded logical qubits (or, more generally, qudits), given by $d_L = 2^k$. In particular, in the limit $d_L \rightarrow \infty$ the necessary and sufficient condition expressed in \eqref{eq:W_code_bound_epsilon} becomes 
\begin{align} \label{eq:eps_Ne_n_asymptotic}
    \varepsilon \ge \frac{N_e}{n} .
\end{align}
However, the dimension of each individual physical subsystem will scale linearly with the dimension of the logical space. We highlight that the scaling in \eqref{eq:eps_Ne_n_asymptotic} is in agreement with the $\Omega(1/n)$ scaling found in the previous work~\cite{Faist2020ContinuousSymmQEC}.

In \figref{fig:Wstate} we plot the performance of this code for different values of $N_e$ in the limit of arbitrary numbers of encoded logical qubits. For example, with $n=100$ physical qutrits we can encode arbitrary many logical qubits in the $W$-code, which admits a universal transversal gateset and can correct for the erasure of a single physical subsystem with error $\varepsilon = 0.01$.  In Ref.~\cite{Faist2020ContinuousSymmQEC}, the authors provide the following upper bound on error associated with decoding the $W$-state code quantified via the worst-case entanglement fidelity~\cite{Schumacher1996Entanglement,Gilchrist2005Distance} $F_{\mathrm{EF}}$ between the channels $(\D \circ \N \circ \E)_L$ and $\mathrm{id}_L$, for the case of $N_e =1$: 
\begin{align}
   1 - F_{\mathrm{EF}} \le \frac{\sqrt{2} + d_L}{\sqrt{n}}.
\end{align}
Upper bounds on $ 1 - F_{\mathrm{EF}}$ also form upper bounds on 
 the definition of worst-case error $\varepsilon$ considered here in \eqref{eq:epsilon_correctability} (in terms of worst-case error of decoding with respect to all pure states of the logical system). Therefore, in \figref{fig:Wstate} we also compare our necessary and sufficient condition to this sufficient condition.

We now construct an explicit decoder for the known erasure of $N_e$ physical subsystems in the $W$-state code. This known erasure channel can be represented via
\begin{align} \label{eq:known_erasure}
    \N_{\bm{s}} \coloneqq  \ketbra{\bm{s}}_{X} \otimes \ketbra{0}_{P_{\bms}} \circ \tr_{P_{\bms}},
\end{align}
where $\ketbra{0}_{P_{\bms}}$ is some fixed constant state which is prepared over the erased subsystems $P_{\bms}$ and which we shall take to be the computational basis all zeros state, and  $\bm{s} \in \{0,1\}^n$ is a bit string with $1$'s in positions labeling which of the $n$ physical subsystems where erased. For example, $\tr_{P_{10\dots 0}} \coloneqq \tr_{P_1}$ denotes the erasure of the first physical subsystem $P_1$, and we subsequently record this information in the state $\ket{10\dots0}_X$ of the classical register $X$. Moreover, we can identify the number of erased qubits we have $N_e = \abs{\bm{s}}$ with the weight of the bit string $\bm{s}$. For the noise channel in \eqref{eq:known_erasure} and the W-state encoder in \eqref{eq:W_state_encoder} we can construct the following explicit decoder:
\begin{align}
\D_{XP \rightarrow L}(\cdot) &\coloneqq \sum_{\bm{s} \in \{0,1\}^n}  K^{\bm{s}} (\cdot)  K^{\bm{s}\dagger}, \quad   K^{\bm{s}} \coloneqq \bra{\bm{s}}\otimes V^{\bm{s}}, \notag \\
V^{\bm{s}}_{P \rightarrow L} &\coloneqq \sum_j \bra{j}_{P_{\bm{s}}} \otimes  V^{(n - \abs{\bm{s}}) \dagger}_{P_{/\bm{s} \rightarrow L}} ,\label{eq:decoder_W_state}
\end{align}
where $\{\ket{j}_{P_{\bm{s}}}\}$ form an orthonormal basis for the space $P_{\bm{s}}$ and where $V^{(n)}$ is defined in \eqref{eq:W_state_encoder}.
For the sake of being explicit this notation is somewhat hard to parse, so let us unpack it with a concrete example.

\begin{tcolorbox}[breakable,colback=violet!5!white,colframe=violet!0!white]
\begin{example}[The W-state code]
Let us assume that the state $\ket{\psi}_L$ is encoded in the $W$-state code as 
\begin{equation}
    \ket{\psi}_L \rightarrow \ket{\psi^{(n)}}_{P_1 \dots P_n} ,
\end{equation}over $n$ physical subsystems.
Now, suppose that the first $N_e < n$ subsystems $P_1 \dots P_{N_e}$ are erased and this information is stored in the bit string \begin{equation}
    \bm{x} \coloneqq \underbrace{1 \dots 1}_{N_e} \underbrace{ 0 \dots 0 }_{n-N_e},
\end{equation}in the register $X$. The resulting state after encoding followed by this known erasure channel will be
\begin{align}
  &\ketbra{\bm{x}}_X \otimes \ketbra{0}_{P_1 \dots P_{N_e}} \otimes \notag \\
  &\left( \frac{N_e}{n} \ketbra{\perp}  
 + \left(1- \frac{N_e}{n}\right) \psi^{(n-N_e)} \right)_{P_{N_e +1}\dots P_n} \label{eq:noisy_W}
\end{align}
where we have defined $\ket{\perp} \coloneqq \ket{d_L, \dots, d_L}$.

The decoder in \eqref{eq:decoder_W_state} reads-out the information stored in the classical register regarding the locations of the erased subsystem, and applies the conditional isometry 
\begin{align}
    V^{\bm{x}}_{P \rightarrow L} = &\sum_j \bra{j}_{P_1 \dots P_{N_e}} \notag \\
    &\otimes \sum_{i} \ket{i}_L \bra{i^{(n-N_e)}}_{P_{N_e+1} \dots P_n}.
\end{align}
We can readily verify that 
\begin{align}
     V^{\bm{x}\dagger}  V^{\bm{x}} = \id_{P_1 \dots P_{N_e}} \otimes \Pi^{(n- N_e)}_{P_{N_e+1} \dots P_n},
\end{align}
where $\Pi^{(n- N_e)}$ is the projector onto the code subspace of the $W$-code $\E^{(n-N_e)}$ with $(n-N_e)$ physical subsystems.
The result of applying the decoder in \eqref{eq:decoder_W_state} to the state in \eqref{eq:noisy_W} is as follows 
\begin{align}
  \D \circ \N^{\bm{x}} \circ &\E^{(n)}(\psi_L) \coloneqq  \notag \\
  &\frac{N_e}{n} \ketbra{\chi}_L  
 + \left(1- \frac{N_e}{n}\right) \psi_L. \label{eq:decoded_noisy_W}
\end{align}
From \eqref{eq:decoded_noisy_W}, we can read-off that the fidelity of this resulting state with respect to the desired state $\ket{\psi}_L$ satisfies 
\begin{align}
    F(\D \circ \N^{\bm{x}} \circ \E^{(n)}(\psi_L), \psi_L ) \ge 1 - \frac{N_e}{n}, 
\end{align}
which achieves the scaling given in Corollary~\ref{corollary:W-code}.
Finally, we highlight that this analysis was independent of the choice of pure state $\ket{\psi}_L$ encoded in the $W$-code, or indeed the dimension $d_L$ of the logical space.
\end{example}
\end{tcolorbox}

\section{Outlook}
\label{sec:outlook}

We have derived a simple entropic constraint on the existence of quantum error correcting codes supporting a  universal transversal set of gates which can approximately correct for local erasure errors. In contrast to the prior works~\cite{Woods2020continuousgroupsof,Faist2020ContinuousSymmQEC,Kubica2021ApproximateEK}, this condition has the benefit of being both necessary and sufficient. To prove our approximate Eastin-Knill theorem, we derived a new result for multi-state purification in the resource theory of asymmetry, which itself is of independent interest. For example, we expect that the results presented in \secref{sec:asymmetry_purification} will find application in a wide range of other settings where we are practically or fundamentally limited to classes of operations which respect a symmetry principle (e.g.~\cite{Dolev2022Gauging,Hall2012Metrology,Giovannetti2004QuantumLimit}).

Beyond this, we hope that this work paves the groundwork for a range of future directions. One potentially fruitful direction would be to extend our results to provide constraints on transversal implementations of non-universal gatesets. For example, it would be interesting to consider necessary and sufficient constraints on transversal encodings of the discrete Clifford group with respect to Clifford group covariant noise channels (e.g. Pauli noise, local erasure noise). This might provide a route to constructing optimal coder-decoder pairs within the magic state injection model~\cite{Bravyi2016CliffordEquiv}. Another interesting subgroup would be to analyze $U(1)$-covariant codes with respect to $U(1)$-covariant noise (e.g. amplitude damping, local erasure). More speculatively, this might allow for a rigorous comparison of the benefits of codes which support transversal implementations of different subgroups of the full unitary group.

In this work, we have considered the setting of exact symmetry principles and approximate quantum error correction. Recent work has explored the setting of approximate symmetry principles~\cite{Liu2023ApproximateSymmetries}, and in this context has shown that there exist fundamental trade-offs between exact symmetry and exact quantum error correction. It would be of interest to see if we can generalize the approximate Eastin-Knill theorem presented here to this setting. One way of ``lifting" a symmetry constraint locally that is readily described within the resource theory setting is by making use of a quantum reference frame~\cite{bartlett2007reference}. This insight has inspired a number of different works which have incorporated quantum reference frames as a way of circumventing the (approximate) Eastin-Knill theorem~\cite{Woods2020continuousgroupsof,Yang2022Optimal,Hayden2021RefFrame}, typically in the asymptotic regime. It would be interesting to see whether an extension of our single-shot analysis can add any new insights in this context. 

Finally, here we have considered the setting where quantum error correction is guaranteed up to some worst-case error for all pure states of the logical system. 
Most generally, given the error correcting code $\E_{L \rightarrow P}$ and the noise channel $\N_{P \rightarrow P'}$, we can ask whether there exists any \textit{superchannel} mapping $(\N \circ \E)_{L \rightarrow P'}$ to the identity channel on the logical system, up to error tolerance $\varepsilon$. An obvious, but considerably more technically challenging, extension would be to develop necessary and sufficient conditions in this more general setting. A possible route to achieving this would be to extend the asymmetry purification theorems presented here to the case of channel purification, by drawing on a plethora of results from the literature on dynamical resources~\cite{Regula2021ChannelDistillation,Fang2022NogoPurificationChannel}.

\section{Acknowledgements}

RA would like to thank Cristina Cirstoiu, David Jennings, Michalis Skotiniotis, Kamil Korzekwa, Philippe Faist, and Ryuji Takagi for insightful comments on various iterations of this draft. We want to acknowledge funding from Ministry for Digital Transformation and of Civil Service of the Spanish Government through projects, PID2021-128970OA-I00 10.13039/501100011033 and QUANTUM ENIA project call - Quantum Spain project, and by the European Union through the Recovery, Transformation and Resilience Plan - NextGenerationEU within the framework of the Digital Spain 2026 Agenda.

\appendix
\renewcommand{\appendixname}{APPENDIX}

\begin{widetext}
\section{CONVEX RESOURCE THEORIES}

A central concern in any resource theory is to quantify the given resource in question. Complete sets of resource monotones allow for the full specification of any resource theory. Let us consider a closed, convex resource theory with free operations $\O$ which are assumed to be some convex subset of all CPTP maps which contains all identity maps. The latter assumption corresponds to the consistency requirement that doing nothing cannot generate resources from nothing.
For any fixed state $\eta \in \D(B)$ we can define the function $F_\eta : \D(A) \rightarrow \mathbb{R}$ via
\begin{align}
    F_{\eta}(\cdot) \coloneqq \max_{\E\in \O(A\rightarrow B)} \tr [\eta_B \E(\cdot)  ] .
\end{align}
This function is known to be monotonically non-increasing under any free operations drawn from the free set $\O(A\rightarrow B)$ (e.g. see Refs.~\cite{Regula2019General,gour2024resources}). Moreover, a known result from the literature~\cite{gour2024resources} is that if we range over all states $\eta \in \D(B)$, then one obtains a \textit{complete} set of monotones for the resource theory in question. Namely, for any pair of quantum states $\rho \in \D(A)$ and $\sigma \in \D(B)$, $\rho \maj \sigma$ if and only if (e.g.~see Theorem~11.1.1 of the recent review~\cite{gour2024resources})
\begin{align}
    F_\eta (\rho) \ge F_\eta(\sigma),
\end{align}
for all $\eta \in \D(B)$. Moreover, when $\eta = \ketbra{\psi}$ is a pure state, the quantity $F_\psi (\rho)$ corresponds to the optimal fidelity with which the state $\psi$ can be distilled from the initial state $\rho$, and hence has been termed the \textit{fidelity of distillation}~\cite{Regula2018CoherenceDistillation}.

\subsection{A single, complete monotone for exact conversion to pure states}

Let us consider the following proposition, which provides a single complete monotone for state transformations under a closed, convex resource theory for which the final state in question is rank-one.

\begin{tcolorbox}[breakable,colback=white!3!white,colframe=black!100!white]
\begin{proposition} \label{prop:single_condition_pure}
Let $\rho_A$ and $\psi_B$ be two quantum states of systems $A$ and $B$, respectively, where $\psi_B \in \mathrm{pure}(B)$. Moreover, suppose that $\O$ is the set of free operations of a closed, convex resource theory. There exists an operation $\E \in \O(A\rightarrow B)$ such that $\E(\rho_A) = \psi_B$ if and only if $F_{\psi_B} (\rho_A) \ge F_{\psi_B} (\psi_B)$.
\end{proposition}
\end{tcolorbox}
\begin{proof}
Since $F_\eta$ are known to be valid resource monotones for any state $\eta \in \D(B)$, we know that if $\rho \maj \psi$ then $F_\psi(\rho) \ge F_\psi(\psi)$. Therefore, it suffices to show the converse.

Let us assume then that $F_\psi(\rho) \ge F_\psi(\psi)$, or more explicitly that
\begin{align} \label{eq:condition_rank_1}
    \max_{\E \in \O(A \rightarrow B)} \tr[\psi_{B} \E(\rho_A) ] \ge \max_{\E \in \O(B \rightarrow B)} \tr[\psi_{B} \E(\psi_B) ].
\end{align}
We begin by noting that straightforwardly we have the following upper bound which holds for any pair of states $\tau \in \D(S)$ and $\psi \in \D(B)$
\begin{align} \label{subeq:pure_opt}
    \max_{\E \in \O(S \rightarrow B)} \tr[ \psi_B \E(\tau_S)]\le 1,
\end{align}
since this is the global maximum of the trace-product $\tr[\rho \sigma] $ of any pair of quantum states $\rho$ and $\sigma$. Now, by definition of a closed, convex resource theory, $\mathrm{id}_{B\rightarrow B} \in \O(B \rightarrow B)$ is a feasible solution to the optimization problem in \eqref{subeq:pure_opt} for $\tau = \ketbra{\psi}$ and $S=B$. Therefore, we have
\begin{align}
  F_{\psi_B}(\psi_B) \ge  \tr[\psi_B \mathrm{id}_B(\psi_B)] = \abs{\braket{\psi}}^2= 1.
\end{align}
Since this is also the global optimum of $F_\psi(\psi)$ we can conclude that 
\begin{align} \label{eq:tr_product_1}
  F_\psi(\psi) =  \max_{\E \in \O(B \rightarrow B)} \tr[\psi \E(\psi) ] = 1.
\end{align}
Combining Eqs~(\ref{eq:condition_rank_1}), (\ref{subeq:pure_opt}), and (\ref{eq:tr_product_1}) implies that 
\begin{align}
  \max_{\E \in \O(A \rightarrow B)} \tr[ \psi_B \E(\rho_A) ]=1,
\end{align}
must hold. It is clear that the objective function $ \tr[\psi \E(\rho) ]$ attains the value $1$ if and only if there exists $\E \in \O(A\rightarrow B)$ such that $\E(\rho) = \psi$, which completes the proof.\end{proof}

The above theorem applies only to the case where we have a pure target state of the transition. However, should we wish to move beyond this setting, we note that it can be straightforwardly leveraged to produce sufficient conditions on arbitrary state transitions whenever the set of free operations $\O$ in question contains the partial trace $\tr_C \in \O$. In particular, given some output state of interest $\sigma_B$, we can always construct any purification $\psi_{BC}$ such that  $\sigma_B = \tr_C \psi_{BC}$. Then, by Proposition~\ref{prop:single_condition_pure}, we are guaranteed that there exists an operation $\E \in \O(A\rightarrow B)$ such that $\E(\rho_A) = \sigma_B$ if $F_\psi (\rho) \ge F_\psi(\psi)$.

\subsection{Approximate conversion}

Whenever there exists a free operation $\E \in \O$ and a state $\sigma $ satisfying $F(\psi,\sigma) \ge 1-\varepsilon$ such that $\E(\rho) = \sigma$, we will write
\begin{align}
    \rho \stackrel{\O}{\rightarrow}_\varepsilon \psi.
\end{align}
With this notation in place, let us note the following proposition, which now provides a single complete monotone for \textit{approximate} state transformations under any closed, convex resource theory for which the final state in question is rank-one.

\begin{tcolorbox}[breakable,colback=white!3!white,colframe=black!100!white]
\begin{proposition} \label{proposition:approx}
Let $0\le \varepsilon \le 1$. Let $\rho_A$ and $\psi_B$ be two quantum states of systems $A$ and $B$, respectively. Moreover, suppose that $\O$ is the set of free operations of a closed, convex resource theory. We have $\rho \stackrel{\O}{\rightarrow}_\varepsilon \psi$ if and only if $F_\psi (\rho) \ge 1- \varepsilon$.
\end{proposition}
\end{tcolorbox}
\begin{proof}
 If   $\rho \stackrel{\O}{\rightarrow}_\varepsilon \psi $ then, by definition, there exists a state $\sigma$ and a channel $\N \in \O$ such that $\N(\rho) = \sigma$ and $F(\psi, \sigma) \ge 1 - \varepsilon$. Therefore, 
 \begin{align}
     F_\psi (\rho) \equiv \max_{\E \in \O} \tr[\psi \E( \rho) ] \ge \tr[\psi\N(\rho)] =\tr[\psi \sigma ] =F(\psi, \sigma) \ge 1- \varepsilon,
 \end{align}
where the first inequality follows from the definition of the maximum, and in the second line we have used the fact that the fidelity $F(\psi, \sigma) = \tr[\psi \sigma]$ for any pure state $\psi$.

To show the converse, let us assume that $F_\psi (\rho) \ge 1- \varepsilon$, where we recall that
\begin{align} \label{subeq:F_psi_rho}
   F_\psi (\rho) =  \max_{\E \in \O} \tr[\psi \E(\rho) ] .
\end{align}
Now let $\E^* \in \O$ be the optimal solution of the optimization problem in \eqref{subeq:F_psi_rho}. Then, 
 since $\E^*\in \O$ is a CPTP map $ \tau \coloneqq \E^*(\rho) $ is a valid quantum state, which satisfies
\begin{align} \label{subeq:F_psi_E_var}
 F(\psi, \tau) = \tr[\psi \E^*(\rho) ]  \ge 1-\varepsilon.
\end{align}
Therefore $\rho \stackrel{\O}{\rightarrow}_\varepsilon  \psi$, as claimed. This completes the proof.\end{proof}

\subsection{Multi-state conversion}

For convenience of the reader, let us briefly recap the setting of multi-state conversion as described in \secref{sec:multi-conversion} of the main text. We seek to generalize the results we have seen so far to the case where we have two (continuous or discrete) ensembles of states $ \bm{\underline{\smash{\rho}}}_A \coloneqq \{\rho_\mu\}_{\mu \in \Lambda}$ and $ \bm{\underline{\smash{\psi}}}_B \coloneqq \{\psi_\mu\}_{\mu \in \Lambda}$, parameterized by some real-valued $\mu$ over the compact set $\Lambda$.  We seek to answer whether or not there exists a single free channel $\E \in \O(A\rightarrow B)$ and some third ensemble of states $\bm{\underline{\smash{\sigma}}}_B \coloneqq \{\sigma_\mu\}_{\mu \in \Lambda}$ such that 
\begin{align}
    \E(\rho_\mu) = \sigma_\mu \text{ and } F(\psi_\mu,\sigma_\mu) \ge 1 - \varepsilon,
\end{align}
for all $\mu \in \Lambda$. Whenever the answer to this question is affirmative, we shall write 
\begin{align}
   \bm{\underline{\smash{\rho}}} \stackrel{\O}{\rightarrow}_\varepsilon \bm{\underline{\smash{\psi}}} .
\end{align}
In this setting we have the following result.

\begin{tcolorbox}[breakable,colback=white!3!white,colframe=black!100!white]
\begin{proposition} 
\label{prop:multi}
Let $ \bm{\underline{\smash{\rho}}}_A \coloneqq \{\rho_{\mu}\}_{\mu \in \Lambda}$ and $ \bm{\underline{\smash{\psi}}}_B \coloneqq \{\psi_{\mu}\}_{\mu \in \Lambda}$, where $\rho_\mu \in \D(A)$ and $\psi_\mu \in \mathrm{pure}(B)$ for all $\mu \in \Lambda$, where $\Lambda$ is a compact set. Moreover, suppose that $\O$ is the set of free operations of a closed, compact, and convex resource theory. We have $ \bm{\underline{\smash{\rho}}} \stackrel{\O}{\rightarrow}_\varepsilon \bm{\underline{\smash{\psi}}}$ if and only if
\begin{align}
\min_{\{p(\mu) \}}   \max_{\E \in \O}   \tr[  \int_\Lambda d\mu \, p(\mu) \psi_\mu^T \otimes \rho_\mu \J_{\E^\dagger} ]  \ge 1 - \varepsilon ,
\end{align}
where the minimization is performed all probability density functions $ p(\mu) $ over $\Lambda$ such that $ \int_\Lambda d\mu \, p(\mu) = 1$ and $p(\mu) \in \mathbb{R}^+$ for all $\mu \in \Lambda$, and where $\J_{\E^\dagger} \coloneqq \id \otimes \E^\dagger \left(\sum_{i,j}\ketbra{ii}{jj} \right)$ is the Choi matrix of the adjoint channel $\E^\dagger$.
\end{proposition}
\end{tcolorbox}

\begin{proof}
To begin, note that $ \bm{\underline{\smash{\rho}}} \stackrel{\O}{\rightarrow}_\varepsilon \bm{\underline{\smash{\psi}}}$ if and only if there exists $\E \in \O $ such that
\begin{align} \label{subeq:trpsimu}
   \tr[\psi_\mu \E(\rho_\mu)] \ge 1 - \varepsilon,
\end{align}
for all $\mu \in \Lambda$. Or, equivalently  $ \bm{\underline{\smash{\rho}}} \stackrel{\O}{\rightarrow}_\varepsilon \bm{\underline{\smash{\psi}}}$ if and only if
\begin{align} \label{subeq:Lambda_distribution_over}
  \int_\Lambda d\mu \, p(\mu) \tr[\psi_\mu \E(\rho_\mu)]   \ge 1 - \varepsilon,
\end{align}
for all probability density functions $ p(\mu) $ over $\Lambda$ such that $ \int_\Lambda d\mu \, p(\mu) = 1$ and $p(\mu) \in \mathbb{R}^+$ for all $\mu \in \Lambda$. To see this equivalence, notice that the dirac delta probability density functions $p_\nu(\mu) \coloneqq \delta (\mu - \nu) $ for all $\nu \in \Lambda$ form the extreme points of the convex set of all distributions $\{p(\mu)\}$ over $\Lambda$. We therefore find that $ \bm{\underline{\smash{\rho}}} \stackrel{\O}{\rightarrow}_\varepsilon \bm{\underline{\smash{\psi}}}$ if and only if
\begin{align}
   \max_{\E \in \O}  \min_{\{p(\mu) \}} \int_\Lambda d\mu \, p(\mu) \tr[\psi_\mu \E(\rho_\mu)] \ge 1 - \varepsilon,
\end{align}
where the minimization is performed over all probability density functions $ p(\mu) $ over $\Lambda$. The set of free operations $\O$ is compact and convex by assumption, and the set of probability measures over the compact set $\Lambda$ is also a compact~\cite{Prokhorov1956Convergence} and convex set. Moreover, the objective function is linear in $\E$ and $p(\mu)$. Therefore, we can invoke Sion's minimax theorem~\cite{Sion1958Minimax} and we arrive at the following equivalent condition
\begin{align}
1 - \varepsilon \le \min_{\{p(\mu) \}}   \max_{\E \in \O}  \int_\Lambda d\mu \, p(\mu) \tr[\psi_\mu \E(\rho_\mu)]  \notag  =\min_{\{p(\mu) \}}   \max_{\E \in \O}   \tr[  \int_\Lambda d\mu \, p(\mu) \psi_\mu^T \otimes \rho_\mu \J_{\E^\dagger} ],
\end{align}
where in the equality we have used the identity $\tr[Y \E(X)] = \tr[Y^T \otimes X \J_{\E^\dagger}]$ and the linearity of the trace operation.
This completes the proof. \end{proof}

\section{PROOFS OF PROPOSITIONS \ref{prop:asymmetry}-\ref{prop:approximate}, AND LEMMA~\ref{lem:multi}}
\label{appx:proofs_of_theorems}

Here we show that we can leverage Propositions~\ref{prop:single_condition_pure}-\ref{prop:multi} to produce Propositions \ref{prop:asymmetry}-\ref{prop:approximate} and \lemref{lem:multi} presented in the main text as immediate corollaries. The only tool we shall need is the fact that, when the set of free operations $\O = \O_G$ corresponds to the set of $G$-covariant channels defined in \eqref{eq:G-cov_channels} of the main text, it is known that~\cite{gour2018quantum,gour2024resources} 
\begin{align}
    \max_{\E \in \O_G(A\rightarrow B)} \tr[\psi_B \E (\rho_A) ] = 2^{- H_{\mathrm{min}}(B|A)_{\Pi_G(\psi^T \otimes \rho)}}.
\end{align}
To give an idea of where this identity comes from, let us reproduce the calculation given in Eq.~(15.223) of Ref.~\cite{gour2024resources}. First, let $\J_\E\coloneqq \mathrm{id} \otimes \E \ketbra{\id}$ is the Choi matrix of the channel $\E$, let $\E^\dagger \coloneqq \{K_i^\dagger\} $ denote the adjoint channel with respect to the channel $\E \coloneqq \{ K_i\}$, and let  $\ket{\id} \coloneqq \sum_i \ket{ii}$. Moreover, let $\mathrm{unital}(A\rightarrow B) $ as the set of all unital channels from $A$ to $B$: those channels $\E_{A \rightarrow B}$ for which $\E(\id_A) = \id_B$). Then we have 
\begin{align}
    \max_{\E \in \O_G(A\rightarrow B)} \tr[\psi_B \E (\rho_A) ] &=  \max_{\E \in \O_G(A\rightarrow B) } \bra{\id} (\mathrm{id} \otimes \E )(\psi_B^T \otimes \rho )\ket{\id} =  \max_{\E \in \O_G(A\rightarrow B) } \tr[ \J_{\E^\dagger} (\psi_B^T \otimes \rho_A )] \notag \\
    &=  \max_{\E \in CPTP(A\rightarrow B) } \tr[\Pi_G( \J_{\E^\dagger}) (\psi_B^T \otimes \rho )] =  \max_{\E \in CPTP(A\rightarrow B) } \tr[ \J_{\E^\dagger} \Pi_G (\psi_B^T \otimes \rho )] \notag \\
    &=  \max_{\E \in \mathrm{unital}(A\rightarrow B) } \tr[ \J_{\E} \Pi_G (\psi_B^T \otimes \rho )] =  \max_{ X \ge 0, \tr_{B} X = \id_A  } \tr[ X_{BA} \Pi_G (\psi_B^T \otimes \rho )] \notag \\
  &= 2^{- H_{\mathrm{min}}(B|A)_{\Pi_G(\psi^T \otimes \rho)}},
\end{align}
where in the final equality we have identified an expression for the dual form of the conditional min-entropy (e.g. see \cite{tomamichel2013thesis}).

Therefore, we can identify the following expression for $F_{\psi_B}(\rho_A)$
\begin{align}
 F_{\psi_B}(\rho_A) = 2^{-H_{\mathrm{min}}(B|A)_{\Pi^G(\psi_B^T \otimes \rho_A)}} . \label{subeq:logF}
\end{align}

\begin{proof}[Proof of Propositions \ref{prop:asymmetry} and \ref{prop:approximate}]

The resource theory of asymmetry with free operations $\O_G$ as defined in \eqref{eq:G-cov_channels} is a closed, convex resource theory.
Therefore, combining \eqref{subeq:logF} with Proposition~\ref{prop:single_condition_pure} and Proposition~\ref{proposition:approx} respectively gives Proposition~\ref{prop:asymmetry} and Proposition~\ref{prop:approximate}, respectively.\end{proof}

\begin{proof}[Proof of \lemref{lem:multi}]

The resource theory of asymmetry with free operations $\O_G$ as defined in \eqref{eq:G-cov_channels} is a closed, compact, and convex resource theory. Therefore Proposition~\ref{prop:multi} applies. In particular, we have $ \bm{\underline{\smash{\rho}}} \stackrel{G}{\rightarrow}_\varepsilon \bm{\underline{\smash{\psi}}}$ if and only if 
\begin{align} \label{subeq:prob_dist_eps}
\min_{\{p(\mu) \}} \max_{\E \in \O_G}   \tr[  X \J_{\E^\dagger} ] \ge 1 - \varepsilon ,
\end{align}
where $X \coloneqq \int_\Lambda d\mu \, p(\mu) \psi_\mu^T \otimes \rho_\mu$. Now we can identify the expression on the LHS of \eqref{subeq:prob_dist_eps} with a minimization over a conditional min-entropy by exploiting the following equivalence
\begin{align}
     \max_{\E \in \O_G}   \tr[  X \J_{\E^\dagger} ] = \max_{\E \, CPTP} \tr[  \Pi_G(X) \J_{\E^\dagger} ] = 2^{-H_{\mathrm{min}}(B|A)_{\Pi_G(X)}}. \label{subeq:hmin_multi}
\end{align}
Substituting \eqref{subeq:hmin_multi} into \eqref{subeq:prob_dist_eps} and taking the logarithm completes the proof. \end{proof}

\section{PROOF OF LEMMA~\ref{lem:covariant_decoder_optimal}}
\label{app:cov_decoder_optimality}

As shown in Ref.~\cite{Zhou2021newperspectives} under $G$-covariant encoding and noise, the optimal decoder can be assumed to also be $G$-covariant. In the following lemma, for completeness, we confirm that this also holds for our precise definition of $\varepsilon$-correcting codes given in \eqref{eq:epsilon_correctability}. 
\begin{lemma}[\cite{Zhou2021newperspectives}] \label{lem:covariant_decoder_optimal}
Given a $G$-covariant code $\E_{L \rightarrow P}$ which is $\varepsilon$-correctable against the $G$-covariant noise channel $\N_{P \rightarrow P'}$, the optimal decoder can always be assumed to be $G$-covariant.
\end{lemma}

\begin{proof}
We follow the proof strategy of Lemma~2 of Ref.~\cite{Zhou2021newperspectives}.
  The code $\E_{L \rightarrow P}$ is $\varepsilon$-correctable against the noise channel $\N_{P \rightarrow P'}$ iff 
  \begin{align} \label{eq:min_R_fR}
      \max_{\R \in CPTP(P'\rightarrow L)}f(\R) \ge 1-\varepsilon,
  \end{align}
  where 
  \begin{align} \label{subeq:fR}
  f(\R) \coloneqq \min_{\psi_L }  F(\R_{P'\rightarrow L} \circ \tilde{\E}_{L \rightarrow P'}(\psi_L),\psi_L),
  \end{align}
  where the minimization is performed over all pure states in $\mathrm{pure}(L)$ and $\tilde{\E}_{L \rightarrow P'} \coloneqq \N_{P \rightarrow P'} \circ \E_{L \rightarrow P}$. By assumption 
 $\E_{L \rightarrow P}$ and $\N_{P \rightarrow P'}$ are $G$-covariant, and therefore $\tilde{\E}_{L \rightarrow P'}$ is also $G$-covariant. Then, letting $\R^*$ be an optimal solution to the minimization problem in \eqref{eq:min_R_fR}, for all $g \in G$ we have
 \begin{align}
     f(\U^{g \dagger}_{L} \circ \R_{P'\rightarrow L}^* \circ \U^g_{P'}) &= \min_{\psi_L} F( \U^{g \dagger}_{L} \circ \R_{P'\rightarrow L}^* \circ \U^g_{P'} \circ \tilde{\E}_{L \rightarrow P'}(\psi_L),\psi_L) \notag \\
      &= \min_{\psi_L} F( \U^{g \dagger}_{L} \circ \R_{P'\rightarrow L}^* \circ  \tilde{\E}_{L \rightarrow P'}(\U_L^g(\psi_L)),\psi_L) \notag \\
       &= \min_{\psi_L} F(   \R_{P'\rightarrow L}^* \circ  \tilde{\E}_{L \rightarrow P'}(\U_L^g(\psi_L)),\U_L^g(\psi_L)) \notag \\
       &= \min_{\psi_L} F(   \R_{P'\rightarrow L}^* \circ  \tilde{\E}_{L \rightarrow P'}(\psi_L),\psi_L) =f(\R_{P'\rightarrow L}^*), \label{subeq:f_invariant_Ug}
 \end{align}
 where in the second equality we use the fact that $\tilde{\E}$ is $G$-covariant and in the fourth equality we invoke the fact that the set of pure states $\mathrm{pure}(L) = \U_L \circ \mathrm{pure}(L)$ is invariant under the application of any unitary channel $\U_L$. 
We can readily verify that $f$ as defined in \eqref{subeq:fR} is concave (which follows from the fact that the fidelity $F$ is concave in its first argument). Therefore, given any optimal solution $\R^*$ we can therefore always construct the manifestly $G$-covariant channel 
\begin{align}
    \R^G \coloneqq \int_{G} dg \, \U_L^{g\dagger} \circ \R^* \circ \U_{P'}^{g} ,
\end{align}
such that, by \eqref{subeq:f_invariant_Ug}, $f(\R^G) \ge f(\R^*) \ge 1-\varepsilon$, and therefore $\R^G$ is also an optimal decoder. \end{proof}

\section{PROOF OF THEOREM \ref{thm:Ud_covariant}}
\label{apx:proof_lemma_Ud_cov_encoder}

Let us first present the following lemma.
\begin{lemma}[\cite{Alexander2022Infinitesimal}] \label{lemma:Hmin_identity} Let us define the functional $\Phi_{A|B}(K_{AB}) \coloneqq 2^{-H_{\mathrm{min}}(A|B)_{K}}$. Moreover, let $K_{AB}$ be a positive semidefinite, linear operator of the form $K_{AB} \coloneqq \lambda_1 \id_A \otimes L_B + \lambda_2 M_{AB}$ with $L_B, M_{AB} \ge 0$ and $\lambda_1, \lambda_2 >0$. Then we have
\begin{align}
    \Phi_{A|B}(K_{AB}) = \lambda_1 \tr[L_B] + \lambda_2 \Phi_{A|B}(M_{AB}).
\end{align}
\end{lemma}
\begin{proof}
For the sake of completeness, let us reproduce the proof line given in Lemma~19 of Ref.~\cite{Alexander2022Infinitesimal}. By definition of the conditional min-entropy, we have
\begin{align}
  \Phi_{A|B}(K_{AB}) &\coloneqq \inf_{X \ge 0} \left\{ \tr[X_B] \ | \ \id_A \otimes X_B - K_{AB} \ge 0 \right\}  \notag \\
&=\inf_{X \ge 0} \left\{ \tr[X_B] \ | \ \id_A \otimes ( X_B - \lambda_1  L_B) - \lambda_2 M_{AB} \ge 0 \right\}   \label{subeq:identity_hmin}
\end{align}
Since $\lambda_1 L_B \ge 0$ and $M_{AB} \ge 0$, we have $\id_A \otimes ( X_B - \lambda_1  L_B) - \lambda_2 M_{AB} \ge 0 $ implies $ X_B - \lambda_1  L_B \ge 0 $ which implies $X_B \ge 0$. Therefore, we are free to replace the feasible set over which we perform the optimization in \eqref{subeq:identity_hmin} as follows without effecting the calculation:
\begin{align}
 \Phi_{A|B}(K_{AB})  &=\inf_{ X_B - \lambda_1  L_B \ge 0} \left\{ \tr[X_B] \ | \ \id_A \otimes ( X_B - \lambda_1  L_B) - \lambda_2 M_{AB} \ge 0 \right\}  \notag \\
&=\inf_{ Y_B \ge 0} \left\{ \tr[Y_B + \lambda_1 L_B ] \ | \ \id_A \otimes Y_B - \lambda_2 M_{AB} \ge 0 \right\}  \notag \\
&= \lambda_1 \tr[L_B] + \inf_{ Y_B \ge 0} \left\{ \tr[Y_B  ] \ | \ \id_A \otimes Y_B - \lambda_2 M_{AB} \ge 0 \right\} \notag \\
&= \lambda_1 \tr[L_B] +  \Phi_{A|B}(\lambda_2 M_{AB}) = \lambda_1 \tr[L_B] +\lambda_2   \Phi_{A|B}(M_{AB}),
\end{align}
where in the second equality we have made the substitution $Y_B \coloneqq X_B -\lambda_1 L_B$. This completes the proof.\end{proof}

Now we prove \thmref{thm:Ud_covariant} from the main text, which we reproduce here for convenience. 

\Udcovariant* 
\begin{proof}
Given $U(d)$-covariant encoder $\E_{L \rightarrow P}$ and noise $\N_{P \rightarrow P'}$ according to \lemref{lem:covariant_decoder_optimal} the optimal decoder $\D_{P'}$ can be assumed to also be $U(d)$-covariant. Therefore there $\E$ is $\varepsilon$-correctable with respect to $\N$ iff 
\begin{align}
  \bm{\underline{\smash{\rho}}}_{P'} \stackrel{U(d)}{\rightarrow}_\varepsilon \bm{\underline{\smash{\psi}}}_L,
\end{align}
with 
\begin{align}
    \bm{\underline{\smash{\psi}}}_L &\coloneqq \{ U_g \ketbra{\phi}_L U_g^\dagger \ | \ \forall \ U_g \in U(d_L)\}, \text{ and}  \notag \\
    \bm{\underline{\smash{\rho}}}_{P'} &\coloneqq \N_{P \rightarrow P'} \circ \E_{L \rightarrow P} (\bm{\underline{\smash{\psi}}}_L),
\end{align}
where $\ketbra{\phi}_L$ is any fixed pure state of the logical space $L$. Therefore, and since $U(d_L)$ is compact for finite $d_L$, by \lemref{lem:multi}, $\E$ is $\varepsilon$-correctable with respect to $\N$ iff
\begin{align} \label{subeq:necesarry_suff_ph}
\max_{\{p(h)\}} H_{\mathrm{min}}(L|P')_{\Sigma_{LP'}^p} \le \log\frac{1}{1- \varepsilon},
\end{align}
where the minimization is performed all probability density functions $ p(h) $ over $U(d)$ such that $ \int_{U(d)} dh \, p(h) = 1$ and $p(h) \in \mathbb{R}^+$ for all $h \in U(d)$, and where
\begin{align}
\Sigma_{LP'}^p &\coloneqq \Pi^{U(d)}_{LP'} \left( \int_{U(d)} dh \, p(h) \U^h_L(\ketbra{\phi}_L)^T \otimes (\N \circ \E)_{\tilde{L} \rightarrow P'}\circ\U^h_{\tilde{L}}(\ketbra{\phi}_{\tilde{L}}) \right) \notag \\
&=\Pi^{U(d)}_{LP'} \left( \int_{U(d)} dh \, p(h) \bar{\U}^h_L(\ketbra{\phi}_L^T) \otimes (\N \circ \E)_{\tilde{L} \rightarrow P'}\circ\U^h_{\tilde{L}}(\ketbra{\phi}_{\tilde{L}}) \right) \notag \\
&=\int_{U(d)} dh \, p(h) \int_{U(d)} dg (\bar{\U}^g_L \otimes \U^g_{P'}\circ (\N \circ \E)_{\tilde{L} \rightarrow P'}) (\bar{\U}^h_L \otimes \U^h_{\tilde{L}})(\ketbra{\phi}_L^T \otimes \ketbra{\phi}_{\tilde{L}})  \notag \\
&= (\N \circ \E)_{\tilde{L} \rightarrow P'} \circ \int_{U(d)} dh \, p(h) \int_{U(d)} dg (\bar{\U}^g_L \otimes \U^g_{\tilde{L}}) \circ (\bar{\U}^h_L \otimes \U^h_{\tilde{L}})(\ketbra{\phi}_L^T \otimes \ketbra{\phi}_{\tilde{L}})  \notag \\
&= (\N \circ \E)_{\tilde{L} \rightarrow P'} \circ \int_{U(d)} dh \, p(h) \int_{U(d)} dg \, \bar{\U}^{gh}_L \otimes \U^{gh}_{\tilde{L}} (\ketbra{\phi}_L^T \otimes \ketbra{\phi}_{\tilde{L}}) , \label{subeq:integrals_for_days}
\end{align}
where in the fourth equality we have used the fact that $(\N \circ \E)_{\tilde{L} \rightarrow P'}$ is $U(d)$-covariant, since both $\N$ and $\E$ are themselves $U(d)$-covariant.
Now by the invariance of the uniform Haar measure $dg$, for all fixed $h \in U(d)$ we have 
\begin{align} \label{subeq:integrals_integrals}
    \int_{U(d)} dg  (\bar{\U}^{gh}_L \otimes \U^{gh}_{\tilde{L}})(\ketbra{\phi}_L^T \otimes \ketbra{\phi}_{\tilde{L}})  &= \int_{U(d)} dg'  (\bar{\U}^{g'}_L \otimes \U^{g'}_{\tilde{L}})(\ketbra{\phi}_L^T \otimes \ketbra{\phi}_{\tilde{L}}) \notag \\
    &=\lambda \frac{\id_{L\tilde{L}}}{d_L^2} + (1-\lambda) \frac{\ketbra{\id}_{L\tilde{L}}}{d_L},
\end{align}
where the final equality is proven in Example~49 of the review \cite{Mele2024introductiontohaar}, with 
\begin{align}\lambda =d_L^2 \frac{1-d_L^{-1}\bra{\id}\ketbra{\phi}_L^T \otimes \ketbra{\phi}_{\tilde{L}} \ket{\id}}{d_L^2-1}  . 
\end{align}
Now we have
\begin{align}
    \bra{\id}\phi_L^T \otimes \phi_{\tilde{L}} \ket{\id}= \sum_{i,j} \bra{i} \phi_L^T \ket{j} \bra{i} \phi_{\tilde{L}} \ket{j} = \sum_{i,j} \bra{j} \phi_L \ket{i} \bra{i} \phi_{\tilde{L}} \ket{j}  = \abs{\braket{\phi}}^2=1 ,
\end{align}
and therefore 
\begin{align}
    \lambda = \frac{d_L^2 - d_L}{d_L^2 -1},
\end{align}
independently of which state $U^h_L\ket{\phi}_L$ of the logical system was chosen. Importantly, this implies that the expression appearing in \eqref{subeq:integrals_integrals} is independent of $h \in U(d)$. Therefore, substituting \eqref{subeq:integrals_integrals} into \eqref{subeq:integrals_for_days} we obtain for every probability density function $p(h)$:
\begin{align}
\Sigma_{LP'}^p &\coloneqq (\N \circ \E)_{\tilde{L} \rightarrow P'} \circ  \int_{U(d)} dh \, p(h)  \left( \lambda \frac{\id_{L\tilde{L}}}{d_L^2} + (1-\lambda) \frac{\ketbra{\id}_{L\tilde{L}}}{d_L} \right)  \notag \\
&=(\N \circ \E)_{\tilde{L} \rightarrow P'}  \left(\lambda \frac{\id_{L\tilde{L}}}{d_L^2} + (1-\lambda) \frac{\ketbra{\id}_{L\tilde{L}}}{d_L}\right), \label{subeq:sig_lp}
\end{align}
where we have exploited the fact that $\int_{U(d)} dh \, p(h)=1$. Combining Eqs.~(\ref{subeq:necesarry_suff_ph}) and (\ref{subeq:sig_lp}) we find that $\E$ is $\varepsilon$-correctable with respect to $\N$ iff
\begin{align} \label{subeq:hmin_intermediate}
 H_{\mathrm{min}}(L|P')_{\Sigma_{LP'}} \le \log\frac{1}{1- \varepsilon}, \text{ where } \Sigma_{LP'} = (\N \circ \E)_{\tilde{L} \rightarrow P'} \left(\lambda \frac{\id_{L\tilde{L}}}{d_L^2} + (1-\lambda) \frac{\ketbra{\id}_{L\tilde{L}}}{d_L} \right).
\end{align}
Now, let us recall the definition of the functional $\Phi_{A|B}(K_{AB}) \coloneqq 2^{-H_{\mathrm{min}}(A|B)_{K}}$. Applying the function $2^{- (\cdot)}$ to both sides of \eqref{subeq:hmin_intermediate}, we find that the encoder $\E$ is $\varepsilon$-correctable under the noise $\N_{P \rightarrow P'}$ if and only if
\begin{align} \label{subeq:Ud_covariant_condition}
\Phi_{L|P'}\left(\Sigma_{LP'} \right) \ge 1- \varepsilon, \text{ where } \Sigma_{LP'} \coloneqq \frac{\lambda}{d_L^2} \id_L \otimes (\N \circ \E)_{\tilde{L} \rightarrow P'}(\id_{\tilde{L}})+ (1-\lambda) J_{(\N \circ \E)_{\tilde{L} \rightarrow P'}}.
\end{align}
Invoking \lemref{lemma:Hmin_identity}, the LHS of \eqref{subeq:Ud_covariant_condition} can be equivalently expressed as
\begin{align} \label{subeq:phi_lp_sigma}
\Phi_{L|P'}(\Sigma_{LP'}) &= \frac{\lambda}{d_L^2} \tr[(\N \circ \E)_{\tilde{L} \rightarrow P'}(\id_{\tilde{L}})]+ (1-\lambda) \Phi_{L|P'}(J_{(\N \circ \E)_{\tilde{L} \rightarrow P'}}) \notag \\
&=\frac{\lambda}{d_L^2} \tr[\id_{\tilde{L}}]+ (1-\lambda) \Phi_{L|P'}(J_{(\N \circ \E)_{\tilde{L} \rightarrow P'}}) = \frac{\lambda}{d_L}+ (1-\lambda) \Phi_{L|P'}(J_{(\N \circ \E)_{\tilde{L} \rightarrow P'}})
\end{align}
Combining Eqs.~(\ref{subeq:Ud_covariant_condition}) and (\ref{subeq:phi_lp_sigma}), we find that  $\E$ is $\varepsilon$-correctable under $\N$ if and only if
\begin{align}
\frac{\lambda}{d_L}   + (1-\lambda) \Phi_{L|P'}(J_{(\N \circ \E)_{\tilde{L} \rightarrow P'}}) \ge 1- \varepsilon
\end{align}
which rearranges to 
\begin{align}
    H_{\mathrm{min}}(L|P')_{J_{\N \circ \E}} \le  - \log d_L (1-c \, \varepsilon) ,
\end{align} 
where $c = (d_L+1)/d_L$, which completes the proof.\end{proof}

\section{PROOF OF COROLLARY~\ref{corollary:W-code}}
\label{Appx:proof_corollary_W_code}
For convenience, let us reproduce Corollary~\ref{corollary:W-code} as stated in the main text. 
\Wcode* 
\begin{proof}
The Choi state of the $W$-state encoder $\E^{(n)}$ on $n$ physical qubits as defined in \eqref{eq:W_state_encoder} is given by
\begin{align}
    J(\E^{(n)}) \coloneqq \frac{1}{d_L} \sum_{i,j=0}^{d_L-1} \ketbra{i}{j}_L \otimes \ketbra{i^{(n)}}{j^{(n)}}_{P_1 \dots P_n},
\end{align}
where 
\begin{align}
    \ket{i^{(n)}}_{P_1 \dots P_n} = \frac{1}{\sqrt{n}} (\ket{i,d_L, \cdots ,d_L } + \ket{d_L,i,  \cdots ,d_L } +\dots  + \ket{d_L , \dots,d_L ,  i }),
\end{align}
for $i \in \{0,\dots, d_L -1 \}$, are the codewords of $\E^{(n)}$. 
This encoding is symmetric with respect to each of the $n$ physical subsystems. Therefore tracing out any of the $n$ qubits will give the same state. Without loss of generality then, here we shall compute the Choi state of the encoder followed by erasure of the first $N_e$ physical subsystems. First, however, let us compute the Choi state following erasure of the first physical subsystem $P_1$:
\begin{align}
    J(\tr_{P_1}\circ \E^{(n)}) \coloneqq \frac{1}{d_L} \sum_{i,j=0}^{d_L-1} \ketbra{i}{j}_L \otimes \tr_{P_1} \ketbra{i^{(n)}}{j^{(n)}}. \label{subeq:J1}
\end{align}
Now we have
\begin{align}
     \tr_{P_1} \ketbra{i^{(n)}}{j^{(n)}} = \sum_{k=0}^{d_L} \bra{k}_{P_1} \ket{i^{(n)}}_{P_1 \dots P_n}\bra{j^{(n)}} \ket{k}_{P_1},\label{subeq:J2}
\end{align}
where,
\begin{align}
     \bra{k}_{P_1}\ket{i^{(n)}}_{P_1 \dots P_n}  &= \frac{1}{\sqrt{n}} ( \bra{k}_{P_1}\ket{i,d_L, \cdots ,d_L } + \bra{k}_{P_1}\ket{d_L,i,  \cdots ,d_L } +\dots  + \bra{k}_{P_1}\ket{d_L , \dots,d_L ,  i }) \notag \\
 &= \frac{1}{\sqrt{n}} (\delta_{k,i} \ket{d_L ,\dots, d_L} + \delta_{k, d_L} \sqrt{n-1} \ket{i^{(n-1)}}. \label{subeq:J3}
\end{align}
Therefore, combining Eqs.~(\ref{subeq:J2}) and (\ref{subeq:J3}), we obtain  
\begin{align}
     \tr_{P_1} \ketbra{i^{(n)}}{j^{(n)}} = \frac{1}{n}\delta_{i,j} \ketbra{d_L, \dots , d_L}{d_L,\dots, d_L} + \frac{n-1}{n} \ketbra{i^{(n-1)}}{j^{(n-1)}}. \label{subeq:J4}
\end{align}
Substituting \eqref{subeq:J4} into \eqref{subeq:J1}, we arrive at
\begin{align}
    J(\tr_{P_1}\circ \E^{(n)}) &= \frac{1}{n d_L} \sum_{i,j=0}^{d_L-1} \ketbra{i}{j}_L \otimes \left(\delta_{i,j} \ketbra{d_L, \dots , d_L}{d_L,\dots, d_L} + (n-1) \ketbra{i^{(n-1)}}{j^{(n-1)}}\right) \notag \\
    &= \frac{1}{d_L n} \id_L \otimes  \ketbra{d_L, \dots , d_L}{d_L,\dots, d_L} +  \left(1 - \frac{1}{n} \right)  J(\E^{(n-1)}).
\end{align}
By recursion, we find that the erasure of the first $N_e$ qubits represented by $\tr_{ P_{1} ,\dots , P_{N_e}}$ can be written  
\begin{align}
    J(\tr_{ P_{1} ,\dots , P_{N_e}}\circ \E^{(n)}) = \frac{N_e}{d_L n} \id_L \otimes  \ketbra{d_L, \dots , d_L}{d_L,\dots, d_L}_{P_{N_e  +1} \dots P_n} + \left(1 - \frac{N_e}{n} \right) J(\E^{(n-N_e)})_{LP_{N_e +1} \dots P_n}.
\end{align}

Invoking \lemref{lemma:Hmin_identity} we find that 
\begin{align}
    \Phi_{L|P_{N_e +1} \dots P_n}[  J(\tr_{ P_{1} ,\dots , P_{N_e}}\circ \E^{(n)}) ] &= \frac{N_e}{d_L n} + \left(1 - \frac{N_e}{n} \right)  \Phi_{L|P_{N_e +1} \dots P_n} \left[  J(\E^{(n-N_e)})\right] \notag \\ 
    &=\frac{N_e}{d_L n} + \left(1 - \frac{N_e}{n} \right) \Phi_{L|P_{N_e +1} \dots P_n}\left[  \mathrm{id}_L \otimes \E^{(n-N_e)}_{\tilde{L}\rightarrow P} \left(\ketbra{\phi^+}_{L \tilde{L}}\right) \right] \notag \\
    &=\frac{N_e}{d_L n} + \left(1 - \frac{N_e}{n} \right) d_L .
\end{align}
Now since $\E_{L \rightarrow P}$ admits a transversal implementation of the full unitary group, by Corollary~\ref{cor:approx_eastin_knill}, $\E_{L \rightarrow P}$ is $\varepsilon$-correctable with respect to single subsystem erasure if and only if 
\begin{align} \label{subeq_params_Ne}
 \frac{N_e}{d_L n} + \left(1 - \frac{N_e}{n} \right) d_L \ge d_L (1 - c\varepsilon).
\end{align}
Finally, we can rearrange \eqref{subeq_params_Ne} to obtain the following equivalent condition
\begin{align}
    \varepsilon \ge \frac{N_e}{n} \left( 1 - \frac{1}{d_L}\right),
\end{align}
which completes the proof. \end{proof}

\end{widetext}

\bibliographystyle{apsrev4-2}
\bibliography{bibliog}
\end{document}